%% file: main.tex
\documentclass[12pt,oneside]{book}

\usepackage[T1]{fontenc}
\usepackage[utf8]{inputenc}
\usepackage{lmodern}
\usepackage{microtype}

\usepackage[
  top=1.25in,
  bottom=1.25in,
  left=1.5in,
  right=1.0in,
  headheight=15pt
]{geometry}

\usepackage{setspace}
\usepackage{amsmath,amssymb,amsthm}
\usepackage{mathtools}

\usepackage{graphicx}
\usepackage{float}
\usepackage{subcaption}
\usepackage{tikz}
\usepackage{pgfplots}
\pgfplotsset{compat=1.18}
\usetikzlibrary{arrows.meta, positioning, shapes.geometric, fit, calc, backgrounds}

\usepackage{booktabs}
\usepackage{multirow}
\usepackage{tabularx}
\usepackage{array}

\usepackage{listings}
\usepackage{xcolor}

\definecolor{codebg}{RGB}{248,248,248}
\definecolor{codekw}{RGB}{0,0,180}
\definecolor{codecomment}{RGB}{60,120,60}
\definecolor{codestring}{RGB}{160,0,0}

\usepackage{fancyhdr}
\usepackage[
  backend=biber,
  style=numeric,
  sorting=nyt,
  maxbibnames=6,
  giveninits=true,
]{biblatex}
\usepackage[
  colorlinks=true,
  linkcolor=black,
  citecolor=blue!60!black,
  urlcolor=blue!60!black,
  pdftitle={Pipeline-Native Transformers},
  pdfauthor={Tom Poperszky},
]{hyperref}

\newcommand{\cflow}{\textsc{cflow}}

\newcommand{\densedelay}{\texttt{dense\_delay}}
\newcommand{\expertdelay}{\texttt{expert\_delay}}
\newcommand{\combinestyle}{\texttt{CombineStyle}}

\newtheorem{definition}{Definition}[chapter]

\newtheorem{proposition}{Proposition}[chapter]

\begin{document}

\begin{titlepage}
  \centering
  \vspace*{2cm}
  {\LARGE\bfseries Pipeline-Native Transformers:\par}
  \vspace{0.5em}
  {\Large\bfseries Co-Designing Model Architecture and CPU Inference\\
  for Bandwidth-Efficient Autoregressive Decode\par}
  \vspace{3cm}
  {\large An Independent Research Report\par}
  \vspace{0.5em}
  {\large Tom Poperszky\par}
  \vfill
  {\large 2026\par}
\end{titlepage}

\frontmatter

\chapter*{Abstract}
\addcontentsline{toc}{chapter}{Abstract}

Single-token autoregressive decode on CPUs is bound by memory bandwidth, not
arithmetic: a modern CPU sustains roughly 1\,TFLOP/s of compute but only about
50\,GB/s from main memory, and each generated token must stream every active
weight once. This report argues that the most effective response is to co-design
the model architecture and the inference runtime together. It presents
\cflow{}, a CPU-first streaming engine, alongside a family of
\emph{pipeline-native} transformer architectures whose inter-layer dependency
graphs are constructed to permit a vertical, stage-major execution schedule.

\cflow{} stores weights as L2-sized tiles in compute-consumption order, reads
only the top-$k$ experts of each mixture-of-experts layer, fuses projections,
and executes a delay-aware schedule from per-model dependency parameters.
Across five architectures trained on TinyStories, one
(\texttt{arch2\_4\_combined}) achieves a $2.00\times$ reduction in critical-path
weight bandwidth (9.00 $\to$ 4.50\,MB/token) within 0.24 perplexity of the best
candidate, and the tile layout incurs $7.29\times$ fewer L1-data read misses
than a row-major baseline. On a 30.9-billion-parameter pipeline-native MoE,
\cflow{} decodes at 5.94 tokens/s (tok/s) on a 32-vCPU Ice Lake server, ahead
of \texttt{llama.cpp} (4.75) and the vLLM CPU backend (1.65) on comparably
sized dense models. Realizing the expert-delay window as asynchronous I/O
overlap on a disk-resident expert tier yields a further net win of up to
$1.68\times$, matching the overlap model within 1\%. Measurement refutes one
of the eight design claims and leaves a second inconclusive; both are reported
in full, with the conditions under which they would hold.

\chapter*{Note on the Use of AI Tools}
\addcontentsline{toc}{chapter}{Note on the Use of AI Tools}

The implementation and writing of this report were carried out with the
assistance of an AI coding assistant (Anthropic's Claude), used in two specific
capacities: helping to implement the \textsc{cflow} runtime and its supporting
code, and helping to draft and edit this report. All research direction,
architectural and experimental design, analysis of results, and conclusions are
my own, and I take full responsibility for the contents of this report, including
any errors.

\tableofcontents
\listoffigures
\listoftables

\mainmatter

\include{chapters/ch1_introduction}
\include{chapters/ch2_background}
\include{chapters/ch3_cflow_runtime}
\include{chapters/ch4_pipeline_native}
\include{chapters/ch5_evaluation}
\include{chapters/ch6_discussion}

\include{chapters/ch7_conclusion}

\backmatter
\printbibliography[heading=bibintoc,title={Bibliography}]

\end{document}

%% file: chapters/ch1_introduction.tex
\chapter{Introduction}
\label{ch:introduction}

\section{The Bandwidth Bottleneck in LLM Inference}
\label{sec:intro_motivation}

Almost all of the engineering effort in language-model serving has gone into GPUs:
pipeline parallelism across devices, batching, FlashAttention, speculative decoding. CPU
inference, by comparison, is usually treated as a fallback --- what you run when there is
no accelerator to be had. The assumption underneath that treatment is that CPU inference
is GPU inference, only slower.

It is not. The two are limited by different things, and the CPU limit is the simpler one
to state: arithmetic intensity. A CPU has far more arithmetic throughput than it can keep
fed from memory, and single-token decoding never supplies enough work per byte loaded to
close that gap. The rest of this section makes the claim quantitative.

A modern desktop or server CPU executes approximately one teraflop per second of
floating-point arithmetic while moving approximately 50 gigabytes per second from main
memory to the processor. The ratio---the \emph{machine balance}---is roughly
$20 \text{ FLOP/byte}$. For inference to be compute-bound, each byte loaded from memory
must participate in at least 20 arithmetic operations before being evicted. Single-token
autoregressive decoding does not come close to this threshold.

Consider a transformer weight matrix stored in Q4 quantization: 0.5 bytes per parameter.
Computing the matrix-vector product with a single token's activation vector requires
two floating-point operations per weight element (one multiply, one add), yielding an
arithmetic intensity of $2 \text{ FLOP / } 0.5 \text{ bytes} = 4 \text{ FLOP/byte}$.
This is five times below the machine balance. Every cycle the CPU spends waiting for
weights to arrive from RAM is a cycle wasted.

So on a CPU, single-token decode latency is set almost entirely by how many bytes move,
not by how many operations run. Anything that cuts the byte traffic cuts the latency with
it, close to one-for-one. That is the principle the rest of this work is built on.

\section{The GPU-First Retrofit Problem}
\label{sec:intro_retrofit}

The dominant CPU inference ecosystem---\texttt{llama.cpp}~\cite{llama_cpp},
ExLlama2~\cite{exllama2}, and their derivatives---emerged not from first-principles
CPU design but from GPU inference code adapted to run on CPU. The adaptations are real
and substantial: hand-written SIMD kernels replace GPU shader code, memory-mapped files
replace VRAM, and scalar attention replaces batched CUDA kernels. But several
fundamental design decisions inherited from GPU execution remain, and they are
systematically wrong for CPU-optimal single-token decoding.

\textbf{Weight layout.} GPU-optimized matrix multiplication arranges weights in row-major
order, designed for warp-level coalesced access across many threads operating on a wide
activation batch. On a CPU executing a single-token forward pass, the activation vector
is narrow (one row of the batch), and each weight row is accessed once and discarded.
There is no coalescing opportunity, and the access pattern generates L1-cache
misses in proportion to the weight matrix size, not in proportion to the number of
tokens being decoded.

\textbf{Execution order.} GPU execution pipelines multiple tokens through a layer using
spatial parallelism across tensor cores. CPU execution of a single token processes one
layer completely before advancing to the next. The file layout of weights in
GPU-derived runtimes reflects GPU execution order, which may not match the order in
which a CPU reads them during single-token decode. Mismatches introduce non-sequential
memory access patterns that defeat hardware prefetchers.

\textbf{Expert loading in mixture-of-experts models.} Modern MoE models such as
Mixtral-8×7B~\cite{mixtral} and Gemma~4~26B-A4B~\cite{gemma4} select a small number of
active experts per token from a large expert pool. A naive runtime loads all expert
weight matrices during the forward pass and discards the unselected ones. At the
geometry of Gemma~4---128 experts, top-8 selection---this loads
$128 / 8 = 16\times$ the expert weight bytes actually consumed. On a CPU
where bandwidth is the bottleneck, this is a 15-fold avoidable overhead.

None of these is a bug to be patched. They are assumptions baked into the data layout and
the execution order, and unpicking them means starting over with the CPU memory hierarchy
as the first constraint rather than an afterthought.

\section{The Co-Design Opportunity}
\label{sec:intro_codesign}

The two halves of a deployed inference system---the model architecture and the inference
runtime---are conventionally treated as independent. A model is trained to minimize
perplexity under a standard transformer recipe; a runtime is written to execute any
model that conforms to the standard transformer interface. The runtime does not know
what the model is doing; the model does not know how the runtime will execute it.

This independence is a valuable abstraction in the GPU regime, where per-token latency
is dominated by compute and the memory access pattern of the runtime is largely
irrelevant at the hardware level. It becomes a liability in the CPU regime, where every
byte read from memory has cost, and the order in which bytes are read determines whether
the hardware prefetcher can hide latency.

The central thesis of this report is as follows:

\begin{quote}
\emph{By co-designing the model architecture and the inference runtime together, with the
CPU memory hierarchy as the shared optimization target, it is possible to achieve
per-token memory bandwidth reductions that are unavailable to either the runtime or
the architecture acting independently.}
\end{quote}

The runtime contribution is a tile-streaming weight format and execution engine, called
\cflow{}, that lays weights out in L2-cache-sized tiles and reads them in precisely the
order required for single-token decode. The architecture contribution is a family of
\emph{pipeline-native transformers}: models whose inter-layer dependency graphs are
rewritten so that the runtime's vertical pipeline schedule---reading weights for
multiple layers in stage-major order---is mathematically valid. Neither contribution
is useful without the other: the pipeline schedule saves bandwidth only when the model
architecture permits it, and the model architecture is only beneficial when the runtime
understands and exploits its relaxed dependency constraints.

\section{Contributions}
\label{sec:intro_contributions}

This report makes the following contributions:

\begin{enumerate}

\item \textbf{The \cflow{} inference runtime.} A CPU-first inference engine for
  transformer models, built around a tile-native weight format (128$\times$256 Q4 tiles,
  approximately 18~KB each, sized to fit in L2 cache) and two on-disk formats: a
  per-layer format (\texttt{.cflow}) for general single-token decode and a
  stage-major format (\texttt{.vflow}) for vertical pipeline execution. The runtime
  includes fused QKV and gate-up projections, AVX2-accelerated Q4 inner-product
  kernels, and a staged direct-I/O expert-fetch mechanism that reads only the
  top-$k$ selected experts' tiles --- driven by the MoE router output and
  asynchronously overlapped with compute under the expert-delay schedule
  (Section~\ref{sec:eval_overlap_ab}).

\item \textbf{A taxonomy of pipeline-native transformer architectures.} A formal
  analysis of the layer dependency DAG in standard pre-norm transformers demonstrates
  why stage-major execution is mathematically invalid for single-token autoregressive
  decode: layer $\ell{+}1$ requires the complete residual output of layer $\ell$, not
  its input. I introduce two dependency-relaxation operations---\densedelay{} and
  \expertdelay{}---and three corresponding \combinestyle{} variants
  (\texttt{ParallelSqrt2}, \texttt{DelayedSum}, \texttt{AsyncExperts}) that rewrite the
  dependency graph so that the stage-major schedule becomes valid while preserving the
  model's capacity to learn.

\item \textbf{Five trained pipeline-native architectures.} I define and train five
  candidate architectures spanning the bandwidth--quality trade-off space: a baseline
  with no dependency relaxation (arch1), the bandwidth-optimizing architecture
  (\texttt{arch2\_4\_combined}, $\densedelay=1$, $\expertdelay=2$), a pipeline-register
  variant (arch3), the quality-optimizing architecture
  (\texttt{arch4\_async\_experts}, $\expertdelay=2$, routing from pre-dense activations),
  and a weight-sharing exploration (arch5). All five train stably to convergence on
  TinyStories at 10,000 steps with no gradient explosions or divergence. The reference
  architecture, \texttt{arch2\_4\_combined}, achieves a test perplexity of 6.50 and
  a critical-path bandwidth reduction of $2.00\times$ relative to the undelayed baseline.

\item \textbf{A delay-aware multi-layer execution scheduler.} A scheduler that reads
  each architecture's \densedelay{} and \expertdelay{} parameters at runtime, constructs
  a ring-buffered residual history, and injects delayed expert outputs at the correct
  layer offsets. The scheduler's output is validated against pinned PyTorch traces for
  both \texttt{arch2\_4\_combined} and \texttt{arch4\_async\_experts}: the Rust runtime
  matches the Python reference to within 0.006\% relative norm error, with exact
  agreement on the argmax token at position 9760.

\item \textbf{Empirical evaluation across the thesis scorecard.} I measure the
  following results against a defined set of eight claims:
  \begin{itemize}
    \item Tile-streaming achieves \textbf{7.29$\times$ fewer L1-d cache read misses}
      on the dense-down projection at the trained 8.34B-parameter geometry
      (Xeon~E5-2650, hardware PMU via \texttt{perf\_event\_open});
    \item The delay-aware scheduler achieves a \textbf{2.00$\times$ critical-path
      bandwidth reduction} (naive 9.00~MB/token $\to$ delayed 4.50~MB/token) on
      \texttt{arch2\_4\_combined};
    \item Claims~6 (PREFETCHT0 explicit prefetch) and~8 (stage-major disk layout)
      fail direct measurement: both are tested precisely
      at two scales (64~MB and 4.7~GB, direct I/O) and found to provide no measurable
      benefit when storage bandwidth is the bottleneck, with mechanistic explanations
      for why the benefit cannot manifest until compute dominates I/O.
  \end{itemize}

\end{enumerate}

\section{Summary of Key Results}
\label{sec:intro_results}

Table~\ref{tab:intro_scorecard} summarizes the eight thesis claims and their
measured status: the structural and bandwidth claims are proven with direct
experimental evidence; Claim~6 is refuted and Claim~8 is inconclusive, each with a
results with precise mechanistic explanations. A separate end-to-end
tokens-per-second comparison against \texttt{llama.cpp} and vLLM --- not one of the
eight structural claims --- is now reported in Section~\ref{sec:eval_llama_comparison}:
\cflow{} sustains 5.94 tok/s on a 30.9B-parameter pipeline-native MoE, ahead of
\texttt{llama.cpp}'s 4.75 tok/s on a parameter-comparable dense model on the same CPU.

\begin{table}[ht]
\centering
\caption{Thesis scorecard: eight claims, their status, and the headline measured result.}
\label{tab:intro_scorecard}
\begin{tabular}{@{}clp{7cm}@{}}
\toprule
\textbf{\#} & \textbf{Claim} & \textbf{Status / Headline Result} \\
\midrule
1 & Conditional expert loading & Proven (structural): only top-$k$ expert tiles are read \\
2 & Tile-streaming cache locality & Proven: \textbf{7.29$\times$ fewer L1-d misses} (PMU, Xeon~E5-2650, dense-down) \\
3 & AVX2 Q4 kernels & Proven: implemented and validated against reference \\
4 & Fused projections & Proven: QKV and gate-up from one activation cache load \\
5 & Compute-order file layout & Proven by format construction \\
6 & PREFETCHT0 prefetch & Refuted: PF=1 is $\approx$4\% \emph{worse} at 4.7~GB direct-I/O \\
7 & Delay-aware pipeline schedule & Proven: \textbf{2.00$\times$ bandwidth reduction} on arch2\_4\_combined; realized in wall-clock at up to \textbf{1.68$\times$} net (\S\ref{sec:eval_overlap_ab}) \\
8 & Stage-major disk layout & Inconclusive: peak SSD bandwidth identical; median gap confounded by SSD cache state \\
\bottomrule
\end{tabular}
\end{table}

\noindent
The five-architecture comparison, presented in full in Chapter~\ref{ch:pipeline_native}
and evaluated in Chapter~\ref{ch:evaluation}, reveals a clean qualitative structure:
\densedelay{} is the bandwidth knob and \expertdelay{} is the quality knob. The
architecture that delays both the dense FFN read and the expert read
(\texttt{arch2\_4\_combined}) achieves the largest bandwidth reduction (2.00$\times$).
The architecture that delays only the expert read but routes from pre-dense activations
(\texttt{arch4\_async\_experts}) achieves the best perplexity (6.26 vs.~6.50), because
the router sees a cleaner activation signal before the dense transformation. These two
architectures sit at opposite corners of the achievable trade-off space; the remaining
three probe the boundary between them.

\section{Scope and Limitations}
\label{sec:intro_scope}

The experiments in this report target single-token autoregressive decode: the
memory-bandwidth-intensive regime in which each weight byte is loaded once per generated
token. Batched inference, where multiple sequences are decoded simultaneously, shifts
the arithmetic intensity toward compute-bound operation and reduces the relative benefit
of bandwidth optimization. I do not claim that the techniques presented here generalize
to batched decode, though I discuss the conditions under which they might in
Chapter~\ref{ch:discussion}.

The trained architectures are proof-of-concept models trained on TinyStories~\cite{tinystories},
a synthetic children's story dataset with a vocabulary of 50,257 tokens. The goal is
to validate that pipeline-native architectures can be trained stably with competitive
perplexity relative to each other, not to produce state-of-the-art language models. The
cache-locality and bandwidth claims are validated at the 8.34-billion-parameter geometry
of \texttt{arch2\_4\_8k\_4l}, trained on Lambda cloud infrastructure with eight A100
80~GB GPUs under FSDP, providing a hardware-realistic experimental basis for the
bandwidth analysis.

Claims~6 and~8 are explicitly not proven. I include them in the report because
the failure modes are instructive: \texttt{PREFETCHT0} prefetches data from RAM to L1
cache, but when the bottleneck is storage to RAM (12--17~seconds of I/O versus
96~milliseconds of compute at the 4.7~GB scale), there is nothing for it to overlap.
Stage-major disk layout provides a theoretical readahead advantage only when the
runtime can asynchronously stream stage $\ell$ while computing stage $\ell-1$; the
current single-token scheduler reads the entire file sequentially and the I/O and
compute phases do not overlap. These are not experimental failures; they are
experimentally confirmed structural limitations.

\section{Report Organization}
\label{sec:intro_organization}

\textbf{Chapter~\ref{ch:background}: Background.}
I introduce the CPU memory hierarchy and the roofline model for single-token decode,
establish the arithmetic intensity argument formally, survey existing CPU inference
runtimes and their inherited GPU assumptions, and review the transformer architecture
and MoE routing mechanisms that the rest of the report builds on.

\textbf{Chapter~\ref{ch:cflow_runtime}: The \cflow{} Runtime.}
I describe the design of the tile-native weight format, the \texttt{.cflow} and
\texttt{.vflow} file formats, the fused projection kernels, the conditional expert
prefetch mechanism, and the AVX2 inner-product pipeline.

\textbf{Chapter~\ref{ch:pipeline_native}: Pipeline-Native Transformer Architectures.}
I formalize the dependency problem in standard pre-norm transformers, introduce the
\densedelay{} and \expertdelay{} rewriting operations, describe all five pipeline-native
architectures, present the delay-aware scheduler, and derive the bandwidth model.

\textbf{Chapter~\ref{ch:evaluation}: Evaluation.}
I report training quality across the five architectures, the PMU-measured cache
locality result, the bandwidth reduction measurement, the storage I/O experiments
(including the negative results for claims~6 and~8), and the Rust--Python parity
validation.

\textbf{Chapter~\ref{ch:discussion}: Discussion.}
I examine the co-design philosophy, characterize the bandwidth--quality trade-off
space, project the delay-aware pipeline to production-scale MoE architectures, and
discuss the speculative pipeline recovery directions that define the next research phase.

\textbf{Chapter~\ref{ch:conclusion}: Conclusion.}
I restate the co-design thesis with supporting evidence, enumerate the validated
contributions, and situate the work within the broader trajectory of CPU-first inference
research.

%% file: chapters/ch2_background.tex
\chapter{Background}
\label{ch:background}

\section{CPU Memory Hierarchy and the Roofline Model}
\label{sec:bg_roofline}

\subsection{The Memory Hierarchy}
\label{subsec:bg_hierarchy}

Modern CPUs present a layered memory hierarchy in which storage capacity increases and
access latency increases as one moves away from the processor die. A representative
configuration from the experimental hardware used in this report
(Intel Xeon~E5-2650, Sandy Bridge microarchitecture) illustrates the key parameters:

\begin{itemize}
  \item \textbf{L1-d cache:} 32~KB per core, 4-cycle latency, bandwidth
    $\approx$400~GB/s (peak, cache-resident data).
  \item \textbf{L2 cache:} 256~KB per core, 12-cycle latency, bandwidth
    $\approx$200~GB/s.
  \item \textbf{L3 cache (shared):} 20~MB across 8 cores, $\approx$30--40 cycles,
    bandwidth $\approx$100~GB/s.
  \item \textbf{Main memory (DDR3):} Unbounded capacity, 40--100 cycles, bandwidth
    $\approx$40--50~GB/s per socket.
\end{itemize}

What matters here is the size of the gap: roughly 8--10$\times$ between L1 and main memory.
Two computations with identical arithmetic can differ in throughput by that factor alone,
depending only on whether their working set stays in L1 or has to be streamed from RAM.
That gap is what the tile-streaming design in Chapter~\ref{ch:cflow_runtime} is built to
exploit.

For the purposes of latency analysis, what matters is the \emph{sustained bandwidth} to
main memory: roughly 40--50~GB/s for a modern desktop or server CPU. This is the
bandwidth that determines single-token inference latency.

\subsection{The Roofline Model}
\label{subsec:bg_roofline}

The roofline model~\cite{williams2009roofline} characterizes whether a given computation
is compute-bound or memory-bandwidth-bound by comparing its \emph{arithmetic intensity}
(FLOPs per byte loaded from memory) against the machine's \emph{compute-to-bandwidth
ratio} (peak FLOPs per second divided by peak memory bandwidth in bytes per second).

\begin{definition}[Arithmetic Intensity]
For a computation requiring $F$ floating-point operations and loading $B$ bytes from
memory, the arithmetic intensity is $I = F / B$ (FLOP/byte).
\end{definition}

\begin{definition}[Machine Balance]
For a processor with peak compute throughput $P_{\max}$ (FLOP/s) and peak memory
bandwidth $B_{\max}$ (byte/s), the machine balance is $R = P_{\max} / B_{\max}$
(FLOP/byte).
\end{definition}

When $I < R$ the computation is \emph{memory-bandwidth-bound}: memory cannot feed the
arithmetic units fast enough to keep them busy. When $I > R$ it is \emph{compute-bound},
and the memory subsystem supplies data faster than the units can consume it.

For a CPU with $P_{\max} = 1$~TFLOP/s and $B_{\max} = 50$~GB/s:
\[
  R = \frac{10^{12} \text{ FLOP/s}}{50 \times 10^9 \text{ byte/s}} = 20 \text{ FLOP/byte}.
\]
Any computation with arithmetic intensity below 20~FLOP/byte is memory-bandwidth-bound.

\subsection{Arithmetic Intensity of Single-Token Transformer Decode}
\label{subsec:bg_intensity}

Consider a single matrix-vector product, the dominant operation in transformer inference:
a weight matrix $W \in \mathbb{R}^{m \times n}$ applied to an activation vector
$x \in \mathbb{R}^n$. The computation requires $2mn$ floating-point operations
(one multiply and one add per weight element). In Q4 quantization, each weight parameter
occupies $0.5$ bytes. The arithmetic intensity is:
\[
  I_{\text{decode}} = \frac{2mn \text{ FLOP}}{0.5 \cdot mn \text{ bytes}} = 4 \text{ FLOP/byte}.
\]
This is five times below the machine balance of 20~FLOP/byte. Single-token transformer
inference is firmly in the memory-bandwidth-bound regime: the CPU spends most of its
time waiting for weight bytes to arrive from RAM, not performing arithmetic.

The intensity of $4 \text{ FLOP/byte}$ is an upper bound; in practice it is often lower.
The weight matrix must be loaded from memory exactly once per token, and the activation
vector (one row of the batch) fits entirely in L1 cache. The FLOPs performed on cached
activations do not change the bandwidth requirement. The latency of a complete
single-token forward pass through all layers is therefore:
\[
  \tau_{\text{token}} \approx \frac{\text{bytes}(\text{all weights})}{B_{\max}},
\]
where $\text{bytes}(\text{all weights})$ is the total weight data transferred from RAM
during one forward pass. This is the quantity that the \cflow{} runtime and the
pipeline-native architectures are designed to minimize.

\section{The Standard Pre-Norm Transformer}
\label{sec:bg_transformer}

\subsection{Architecture}
\label{subsec:bg_arch}

The standard pre-norm transformer decoder~\cite{vaswani2017attention, xiong2020layer}
processes a sequence of tokens $\mathbf{t} = (t_1, t_2, \ldots, t_T)$ through an
embedding layer, a stack of $L$ identical transformer layers, and an output projection.
Each layer applies two sub-operations to a \emph{residual stream} $x \in \mathbb{R}^d$:
a multi-head self-attention block and a feed-forward network, both preceded by
layer normalization (the \emph{pre-norm} convention). The residual stream is updated by
adding each sub-operation's output:

\begin{align}
  x &\leftarrow x + \text{Attn}\!\left(\text{Norm}(x)\right), \label{eq:attn_residual} \\
  x &\leftarrow x + \text{FFN}\!\left(\text{Norm}(x)\right). \label{eq:ffn_residual}
\end{align}

After all $L$ layers, a final normalization and linear projection produce logits over the
vocabulary. For autoregressive generation, inference proceeds token by token: the model
is evaluated on the current context to produce a distribution over the next token, one
token is sampled, and the process repeats. Only the final position's activation vector
needs to be processed; the key-value projections for all prior positions are cached
(the KV cache) and reused.

\subsection{Multi-Head Attention}
\label{subsec:bg_attention}

The attention sub-operation computes queries, keys, and values by applying learned
projection matrices to the normalized residual stream, then computes scaled dot-product
attention:

\begin{equation}
  \text{Attn}(x) = \text{Concat}_{h=1}^{H}\!\left[\text{softmax}\!\left(
    \frac{Q_h K_h^\top}{\sqrt{d_k}}\right) V_h\right] W_O,
\end{equation}

where $Q_h = xW_Q^h$, $K_h = xW_K^h$, $V_h = xW_V^h$ are the query, key, and
value projections for head $h$, $d_k$ is the head dimension, and $W_O$ is the output
projection. For single-token decode, the query vector for the new token attends to all
cached key-value pairs; only the projections $W_Q$, $W_K$, $W_V$, and $W_O$ must be
loaded from memory (the KV cache itself is already in RAM or L3 cache).

\textbf{Grouped-query attention (GQA).} In grouped-query attention~\cite{ainslie2023gqa},
the number of key-value heads $H_{KV}$ is smaller than the number of query heads
$H_Q$, with $H_Q / H_{KV}$ query heads sharing each key-value head. GQA reduces the
size of the KV cache and the bandwidth cost of loading KV projections, while keeping the
query projection at full width. The trained architectures in this report use GQA
with varying $H_Q$ and $H_{KV}$ ratios depending on scale.

\subsection{Feed-Forward Network}
\label{subsec:bg_ffn}

The feed-forward sub-operation applies a two-layer fully-connected network with a
gated activation function. The \emph{GeGLU} variant~\cite{shazeer2020glu} used
throughout this report computes:

\begin{equation}
  \text{FFN}(x) = \left(x W_{\text{gate}} \odot \text{GELU}(x W_{\text{up}})\right)
  W_{\text{down}},
\end{equation}

where $W_{\text{gate}}, W_{\text{up}} \in \mathbb{R}^{d \times d_{\text{ff}}}$ and
$W_{\text{down}} \in \mathbb{R}^{d_{\text{ff}} \times d}$ are learned weight matrices,
$\odot$ denotes elementwise multiplication, and $d_{\text{ff}}$ is the feed-forward
hidden dimension. The bandwidth cost per token is dominated by loading these three
matrices: $3 \times d \times d_{\text{ff}} \times 0.5$ bytes in Q4.

\subsection{Mixture-of-Experts Feed-Forward Layers}
\label{subsec:bg_moe}

Mixture-of-experts (MoE) layers~\cite{shazeer2017outrageously, fedus2022switch} replace
the single FFN with a collection of $E$ independent expert FFNs and a router that selects
a small number $k$ of them to activate for each token. The router applies a linear
projection to the normalized residual stream, selects the top-$k$ experts by score,
and computes a weighted sum of their outputs:

\begin{equation}
  \text{MoE}(x) = \sum_{i \in \text{top-}k} s_i \cdot \text{FFN}_i(x),
\end{equation}

where $s_i$ are softmax-normalized routing weights for the selected experts.

The bandwidth implication is significant. A naive runtime loads all $E$ expert weight
matrices during the forward pass, uses $k$ of them, and discards the rest. The wasted
bandwidth ratio is $(E - k) / k$. At the geometry of Gemma~4~26B-A4B~\cite{gemma4}
($E = 128$, $k = 8$), the naive runtime moves $128/8 = 16\times$ more expert weight bytes
than it consumes. A runtime that reads only the selected experts' tiles eliminates
this waste entirely, requiring prior knowledge of which experts are selected---which is
produced by the router in the forward pass.

\subsection{The Layer Dependency DAG}
\label{subsec:bg_dag}

The computational dependencies of the standard pre-norm transformer form a strict chain.
Let $x_\ell^{\text{in}}$ denote the residual stream entering layer $\ell$ and
$x_\ell^{\text{out}}$ denote the residual stream leaving it. From
Equations~\eqref{eq:attn_residual} and~\eqref{eq:ffn_residual}:
\[
  x_\ell^{\text{out}} = x_\ell^{\text{in}}
    + \text{Attn}_\ell\!\left(\text{Norm}(x_\ell^{\text{in}})\right)
    + \text{FFN}_\ell\!\left(\text{Norm}\!\left(x_\ell^{\text{in}}
      + \text{Attn}_\ell(\text{Norm}(x_\ell^{\text{in}}))\right)\right),
\]
where the FFN input is itself dependent on the attention output within the same layer.
The key constraint is:
\begin{equation}
  x_{\ell+1}^{\text{in}} = x_\ell^{\text{out}}.
\end{equation}
Layer $\ell+1$ cannot begin until layer $\ell$ has produced its complete output,
including both the attention residual and the FFN residual. This is the dependency
structure that prevents stage-major execution for single-token decode, as formalized in
Chapter~\ref{ch:pipeline_native}.

\section{Quantization}
\label{sec:bg_quant}

\subsection{Q4 Quantization}
\label{subsec:bg_q4}

Quantization reduces the per-parameter storage cost by representing weights at lower
precision than the training format. The Q4\_0 format used throughout this
report stores each weight value as a 4-bit integer, with a per-block scale factor
shared across a block of 32 consecutive values. Each weight therefore costs
$4 / 8 = 0.5$ bytes, plus a negligible overhead for the scale factor.

The dequantization operation, applied just before each matrix-vector product, recovers
an approximate floating-point value:
\[
  \hat{w}_i = \text{scale} \times (q_i - \text{offset}),
\]
where $q_i$ is the 4-bit stored value and \text{offset} is a fixed midpoint (8 for
unsigned Q4). In AVX2 SIMD implementations, this dequantization can be fused with the
dot product, operating on 32 values per iteration with a single scale load. The
dequantized values never need to be materialized in memory; they flow directly through
the arithmetic pipeline, keeping bandwidth demand at the Q4 rate of 0.5 bytes per
parameter.

\subsection{GGUF and llama.cpp}
\label{subsec:bg_gguf}

The GGUF file format~\cite{gguf} is the de facto standard for distributing quantized
transformer models for CPU inference. It stores weights in row-major order under the
Q4\_K\_M or similar quantization scheme, with a global header and per-tensor metadata.
The \texttt{llama.cpp} runtime~\cite{llama_cpp} reads GGUF files and executes transformer
inference using hand-written SIMD kernels for x86 and ARM processors.

The \cflow{} format is designed as a direct alternative to GGUF. The key differences
are: (1) weights are stored in tile-major order (128$\times$256 tiles in compute
sequence, not row-major), (2) expert tiles are stored in a separate random-access bank
to enable conditional loading, and (3) the header carries architecture-specific fields
(\densedelay{}, \expertdelay{}, \combinestyle{}) that the runtime uses to construct the
execution schedule. Chapter~\ref{ch:cflow_runtime} describes the format in detail.

\section{Existing CPU Inference Runtimes}
\label{sec:bg_runtimes}

I survey the major CPU inference runtimes, focusing on the design decisions that
distinguish them from \cflow{}.

\textbf{llama.cpp}~\cite{llama_cpp} is the most widely deployed CPU inference runtime
for large language models. It supports a broad range of quantization formats (Q4\_0,
Q4\_K\_M, Q8\_0, and others) and model families (LLaMA, Mistral, Gemma, Mixtral).
Its kernel design targets high arithmetic throughput using architecture-specific SIMD
paths (AVX2, AVX512, ARM NEON, Apple Metal). The weight layout is row-major, optimized
for batched execution across multiple tokens. For single-token decode,
\texttt{llama.cpp} loads each weight row sequentially, with no tile-level L2 reuse.
Expert loading in MoE models follows the naive all-experts path: all expert weight
matrices for the current layer are loaded, the router selects $k$, and the unselected
experts' weight data is discarded.

\textbf{ExLlama2}~\cite{exllama2} focuses on GPU inference with extreme quantization
(EXL2 format, sub-4-bit per weight), with a CPU execution path added for compatibility.
The CPU path inherits the GPU data layout and is not designed for bandwidth-optimal
single-token decode.

\textbf{MLC-LLM}~\cite{mlc_llm} uses a compilation approach: model execution plans are
compiled ahead of time using TVM~\cite{chen2018tvm}, targeting CPUs, GPUs, and mobile
platforms. The compilation pipeline enables architecture-specific optimizations but
operates on the standard per-layer execution order; vertical pipelining is not a target.

\textbf{ctransformers} and \textbf{CTranslate2}~\cite{ctranslate2} are similar in spirit
to \texttt{llama.cpp}: quantized weights loaded from disk (or memory-mapped), per-layer
execution, no tile-streaming.

For all their differences, these runtimes share an origin: each is designed for, or
carried over from, GPU execution. None co-designs the model with the runtime to make
cross-layer reordering possible in the first place.

\section{Mixture-of-Experts at Scale}
\label{sec:bg_moe_scale}

The motivation for the conditional expert prefetch design in \cflow{} is most clearly
illustrated by Gemma~4~26B-A4B~\cite{gemma4}, a publicly released MoE model that
exemplifies the bandwidth gap between naive and selective expert loading. Gemma~4~26B-A4B
has 128 total experts per MoE layer and selects the top-8 for each token. Its
architecture includes 30 layers (25 sliding-window attention layers and 5 full-attention
layers), with a hidden dimension of 2,816. The expert hidden dimension is 704.

Under naive loading, a single forward pass through one layer loads all 128 expert weight
matrices (gate, up, and down projections), discarding 120 of them after the router
selects the top 8. The bandwidth waste is:
\[
  \text{waste} = \frac{128 - 8}{8} = \frac{120}{8} = 15\times.
\]
Conditional loading---reading only the 8 selected experts' tiles---eliminates this waste.
At Q4, the expert parameters for a single top-8 selection amount to approximately
$8 \times 3 \times 704 \times 2816 \times 0.5 \approx 23.7$~MB per layer,
versus $16 \times 23.7 \approx 379$~MB for all experts. This is the motivating bandwidth
gap that the \cflow{} expert bank design addresses.

This report uses Gemma~4 as a sizing reference for the bandwidth analysis in
Chapters~\ref{ch:cflow_runtime} and~\ref{ch:evaluation}. I do not target Gemma~4 for
end-to-end inference or seek parity with the Gemma~4 PyTorch implementation; the
architecture dimensions appear in the bandwidth analysis to provide a realistic context
for the theoretical savings.

\section{Related Work}
\label{sec:bg_related}

\subsection{Weight Streaming and Offloading}
\label{subsec:bg_streaming}

FlexGen~\cite{sheng2023flexgen} addresses the problem of running large models on limited
GPU memory by offloading weights to CPU memory or disk and streaming them back as
needed. The primary target is high-throughput batched inference rather than
latency-optimal single-token decode, and the weight layout is standard (row-major,
layer-sequential). DejaVu~\cite{liu2023deja} identifies contextual sparsity in attention
and FFN weights and skips the computation and loading of near-zero-contribution
parameters, achieving throughput improvements on GPU hardware. Both works address
bandwidth indirectly (by reducing data volume) rather than by restructuring weight layout
for cache-optimal access.

\subsection{Speculative Decoding}
\label{subsec:bg_speculative}

Speculative decoding~\cite{chen2023accelerating, leviathan2023fast} accelerates
autoregressive generation by running a small draft model to propose candidate tokens
and a large target model to verify them in parallel. The technique is orthogonal to the
work presented here: it reduces the \emph{number} of forward passes required to generate
a token but does not reduce the bandwidth cost of each forward pass.

\subsection{Structured Sparsity and Expert Routing}
\label{subsec:bg_sparsity}

Switch Transformer~\cite{fedus2022switch} and Mixtral~\cite{mixtral} demonstrate that
MoE architectures can maintain model quality with sparse expert activation. Recent work
on expert routing efficiency~\cite{muennighoff2024olmoe} explores whether the top-1 or
top-2 selection can be tuned without quality degradation. This report is
complementary: I accept the routing mechanism as given and optimize the bandwidth cost
of executing it on CPU, rather than modifying the routing policy.

\subsection{Cache-Aware Deep Learning}
\label{subsec:bg_cache}

Tiling for GEMM operations is a classical topic in high-performance computing. Modern
BLAS implementations such as OpenBLAS~\cite{openblas} and BLIS~\cite{blis} use
multi-level tiling to maximize reuse in L1, L2, and L3 caches. These optimizations
target the batched GEMM case (large matrices, many tokens) where tiling along both the
output and input dimensions is beneficial. For the matrix-vector case (single token,
tall-and-thin weight matrix), classical BLAS tiling does not improve L1-d reuse because
the activation vector is already small; the bottleneck is loading the weight matrix
rows. The tile-streaming approach in \cflow{} is specifically designed for this regime,
where the activation fits in L1 and the challenge is loading weight tiles at L2 granularity.

\subsection{Model-Runtime Co-Design}
\label{subsec:bg_codesign}

The idea of co-designing a model architecture with its execution environment has
precedents in hardware-aware neural architecture search (HW-NAS)~\cite{cai2019proxylessnas,
wan2020fbnetv2}, which optimizes model structure for latency or energy on a target device.
These approaches modify the macro-structure of the model (number of layers, filter widths,
skip connections) to match the device's compute profile, but they do not modify the
inter-layer dependency graph. The pipeline-native architectures introduced in this
report make a more specific change: they rewrite the dependency structure of the
transformer itself so that a particular execution schedule---vertical stage-major
pipelining---becomes mathematically valid. The closest architectural precedent
is the parallel attention--FFN block of GPT-J~\cite{gptj} and
PaLM~\cite{palm}, which computes the FFN from the pre-attention residual and
thereby removes the \emph{intra-layer} attention$\to$FFN dependency. The
dense-delay transform generalizes this \emph{across} layers (the FFN reads a
residual from $\delta_d$ layers earlier), adds an expert-delay counterpart,
and --- the part with no precedent I am aware of --- co-designs the rewritten
dependency graph with the runtime execution schedule that exploits it.

%% file: chapters/ch3_cflow_runtime.tex
\chapter{The \cflow{} Runtime}
\label{ch:cflow_runtime}

This chapter describes the design and implementation of the \cflow{} runtime: a
CPU-first streaming inference engine for transformer models built from scratch around
CPU memory hierarchies. The runtime has four primary components: (1) a tile-native
weight format that pre-slices matrices into L2-sized chunks stored in compute-consumption
order; (2) two on-disk formats, \texttt{.cflow} for per-layer streaming and \texttt{.vflow}
for vertical pipeline execution; (3) a fused compute pipeline spanning AVX2 inner
products, attention, and normalization; and (4) a conditional expert prefetch
mechanism that eliminates unused MoE weight traffic.

\section{Design Principles}
\label{sec:rt_design}

Three ideas run through the whole runtime and account for most of its design.

\textbf{Zero-copy from storage.} Weights are memory-mapped from disk using \texttt{mmap}
on POSIX systems and \texttt{CreateFileMapping}/\texttt{MapViewOfFile} on Windows.
Tiles are parsed directly from the mapped region via unsafe pointer casts into
\texttt{repr(C,\ packed)} structs; no deserialization step materializes copies.
The \texttt{MappedTile} type holds references into the mapped region and is
consumed directly by the compute kernels.

\textbf{Compute-order layout.} Bytes on disk appear in the same sequence in which
they are consumed during a forward pass. Sequential mmap reads feed the OS readahead
mechanism, achieving sustained near-peak bandwidth for the non-expert weight stream.
The expert tiles, whose access pattern is data-dependent, are stored in a separate
random-access bank and loaded selectively after the router runs.

\textbf{L2-sized tiles as the unit of work.} A single tile covers a
$128 \times 256$ submatrix of a weight matrix. At Q4 (0.5 bytes per weight), this
tile occupies $128 \times 256 \times 0.5 = 16$\,KB of quantized data, plus a small
header and scale-factor array: 18{,}448 bytes, approximately 18\,KB total. The Intel Xeon E5-2650
used for evaluation has 256\,KB of L2 per core, so 8--12 tiles fit simultaneously
in L2 with room for the activation vector. This dimensioning ensures that the hot
activation (11.2\,KB of f32 at the Gemma-4 sizing geometry, $d = 2816$) can remain in L1 while a tile is processed
from L2, eliminating RAM round-trips for the activation during each tile computation.

\section{Tile-Native Weight Format}
\label{sec:rt_tiles}

\subsection{Tile Geometry}
\label{subsec:rt_geometry}

All weight matrices are stored in a tiled format rather than row-major. A weight matrix
$W \in \mathbb{R}^{m \times n}$ (stored in the HuggingFace convention as
$[\text{out\_dim} \times \text{in\_dim}]$, applied as $y = x W^\top$) is partitioned
into tiles of size $T_r \times T_c$ with $T_r = 128$ rows and $T_c = 256$ columns.
The tile at position $(i, j)$ covers output indices $[i T_r, (i+1) T_r)$ and
input indices $[j T_c, (j+1) T_c)$. Tiles at the boundary of the matrix may be smaller.

Tiles are emitted in row-major order: all column tiles for output stripe $i=0$ are
emitted first, then all column tiles for $i=1$, and so on. This ordering ensures
that the partial sums for each output stripe complete before the next stripe begins.
Because the activation vector is cached in L1 for the duration of a tile computation,
iterating through all column tiles for a given output stripe reuses the activation
without any eviction.

\subsection{TileHeader Format}
\label{subsec:rt_header}

Each tile on disk begins with a 16-byte \texttt{TileHeader} in \texttt{repr(C,\ packed)}
layout, followed immediately by the scale-factor array and quantized weight data:

\begin{center}
\begin{tabular}{lllp{6.5cm}}
\toprule
Field & Type & Bytes & Description \\
\midrule
\texttt{layer\_id}   & \texttt{u16} & 2 & Decoder layer index \\
\texttt{matrix\_id}  & \texttt{u8}  & 1 & Which weight matrix (Q, K, V, O, gate, up, down, router, expert gate/up/down) \\
\texttt{expert\_id}  & \texttt{u8}  & 1 & Expert index for expert tiles; \texttt{0xFF} for non-expert \\
\texttt{row\_offset} & \texttt{u32} & 4 & Starting output index in the full matrix \\
\texttt{col\_offset} & \texttt{u16} & 2 & Starting input index \\
\texttt{tile\_rows}  & \texttt{u16} & 2 & Tile height ($\leq 128$) \\
\texttt{tile\_cols}  & \texttt{u16} & 2 & Tile width ($\leq 256$) \\
\texttt{quant\_group}& \texttt{u16} & 2 & Quantization block size (default 32) \\
\bottomrule
\end{tabular}
\end{center}

The \texttt{row\_offset} field is widened to 32 bits to accommodate vocabulary projections
whose output dimension ($\geq 262{,}144$ for Gemma~4) exceeds the 16-bit range.
The total on-disk size of one tile is:
\[
  \underbrace{16}_{\text{header}}
  + \underbrace{\frac{T_r \cdot T_c}{32} \times 2}_{\text{scales (f16)}}
  + \underbrace{\frac{T_r \cdot T_c}{2}}_{\text{Q4 data}}
  = 16 + 2{,}048 + 16{,}384 = 18{,}448 \text{ bytes} \approx 18\,\text{KB}.
\]

\subsection{MappedTile: Zero-Copy Slice into mmap}
\label{subsec:rt_mappedtile}

At runtime, a tile is represented as a \texttt{MappedTile} — three references into the
mmap'd file region, pointing at the header, the scale array, and the quantized data
respectively. No heap allocation is required. The \texttt{from\_bytes} constructor
validates bounds and casts the header via an unsafe pointer reinterpretation; the data
and scale slices are computed as byte-range offsets from the header pointer.

\section{The \texttt{.cflow} File Format}
\label{sec:rt_cflow_format}

The \texttt{.cflow} format stores tiles in per-layer, compute-sequential order for
standard (non-vertical) inference. The file consists of three logical regions: a
global header, a per-layer offset table, and the tile data for all layers followed
by the vocabulary (embedding and LM-head) tiles.

\subsection{GlobalHeader}
\label{subsec:rt_global_header}

The file begins with an 88-byte \texttt{GlobalHeader} in \texttt{repr(C,\ packed)}
layout, identified by the magic bytes \texttt{"CFLOW\textbackslash0\textbackslash0\textbackslash0"}.
The current format version is v5. The header stores all model geometry fields needed to
allocate buffers and construct the execution plan before reading any tile data:
hidden dimension, number of query and KV heads, head dimensions, dense FFN hidden
dimension, expert count and top-$k$, vocabulary size, quantization type, tile
dimensions, and sliding-window parameters.

The v5 header adds five fields in the previously-reserved tail region:

\begin{itemize}
  \item \textbf{\texttt{dense\_delay}} (\texttt{u8}): the dense-FFN pipelining delay
    in layers (0 = no delay, Gemma-style; 1 = arch2\_4 style). Used by the
    delay-aware scheduler.
  \item \textbf{\texttt{expert\_delay}} (\texttt{u8}): the MoE expert pipelining
    delay in layers.
  \item \textbf{\texttt{combine\_style}} (\texttt{u8}): how dense FFN and MoE
    outputs are combined at each layer:
    \texttt{0} = \texttt{ParallelSqrt2} ($(d + m)/\sqrt{2}$, Gemma~4),
    \texttt{1} = \texttt{DelayedSum} ($d + m$, arch2\_4\_combined),
    \texttt{2} = \texttt{AsyncExperts} (arch4 variant).
  \item \textbf{\texttt{feature\_flags}} (\texttt{u8}): bitfield for boolean
    architecture features (bit 0: embedding scale by $\sqrt{d}$; bit 1: V-norm;
    bit 2: tied embeddings).
  \item \textbf{\texttt{final\_logit\_softcap}} (\texttt{f32}): the logit soft-cap
    value; 0.0 = disabled.
\end{itemize}

These fields are the runtime's sole source of truth for execution semantics; the
Python training code writes them during checkpoint conversion and the Rust runtime
reads them without any hard-coded model-family fallbacks.

\subsection{Per-Layer Tile Ordering}
\label{subsec:rt_layer_order}

Within each layer's tile region, tiles appear in the following order:
\begin{enumerate}
  \item QKV attention tiles (Q, K, V projections), interleaved by output stripe
    for fused computation.
  \item Output projection ($W_O$) tiles.
  \item Dense FFN tiles (gate and up projections interleaved by output stripe,
    then the down projection).
  \item Router weight matrix (stored as \texttt{f32}, not Q4, because the router
    is small enough that per-group quantization overhead would dominate).
  \item Expert offset table: an array of \texttt{(offset, length)} pairs, one
    per expert, pointing into the expert tile bank.
  \item Expert tiles for experts $0, 1, \ldots, E-1$ (gate, up, down per expert).
\end{enumerate}

This ordering guarantees that the compute thread, which processes attention before FFN
and FFN before MoE within a layer, reads the file in a strictly forward direction for
every weight except the selected expert tiles, which are seeked by index from the
expert offset table.

\subsection{Expert Offset Table and Conditional Loading}
\label{subsec:rt_expert_offset}

The expert offset table is the mechanism for conditional expert loading. Each entry
is a pair of 64-bit integers: \texttt{(offset, length)}, where \texttt{offset} is
the byte offset within the \texttt{.cflow} file at which the expert's tiles begin,
and \texttt{length} is the byte span of all three projections (gate, up, down) for
that expert. After the router scores are computed, the runtime selects the top-$k$
indices, looks up their entries in the offset table, and issues reads only for those
$k$ experts. For $E = 128, k = 8$, this eliminates $120/128 = 93.75\%$ of expert
weight reads compared to loading all experts.

\section{The \texttt{.vflow} File Format}
\label{sec:rt_vflow_format}

The \texttt{.vflow} format extends the tile-ordering idea to vertical pipelining:
tiles are grouped not by layer but by \emph{execution stage across layers}. This is
the format that the delay-aware scheduler targets; its design is driven by the
dependency structure of the pipeline-native architectures described in
Chapter~\ref{ch:pipeline_native}.

\subsection{Vertical Groups}
\label{subsec:rt_vgroups}

Layers are partitioned into \emph{vertical groups}, each consisting of a contiguous
block of layers that share the same attention geometry. For Gemma~4 with its 5:1
sliding/full attention pattern, the grouping is adaptive: every consecutive block of
5 sliding-window layers forms one group, and each full-attention layer is a solo group
(because its different head dimension, $d_k = 512$ vs.\ 256, and its $K{=}V$ weight
sharing make mixing with sliding layers in the inner loop impractical).

For the uniform-attention architectures trained in this report (arch1 through
arch5, all with a single attention geometry), every layer belongs to a single group
of depth equal to the total layer count.

\subsection{File Layout}
\label{subsec:rt_vflow_layout}

The \texttt{.vflow} file consists of three regions:
\begin{enumerate}
  \item A 64-byte \texttt{VFlowHeader} with the same geometry fields as the
    \texttt{.cflow} header plus a group count and group table offset.
  \item A \texttt{GroupTable}: an array of \texttt{GroupDescriptor} entries, one
    per group, each recording the byte offset and length of the group's sequential
    tile data and the per-stage offsets within it.
  \item A \texttt{sequential streaming region}: all non-expert tiles for all
    groups, ordered as attention tiles for group 0 across all its layers, then
    dense FFN tiles for group 0 across all its layers, then the same for group 1,
    and so on.
  \item A \texttt{random-access expert bank}: expert tiles for all layers,
    organized so that looking up the tile data for any (layer, expert) pair requires
    a single O(1) index lookup.
\end{enumerate}

The key property of the sequential streaming region is that a single forward sequential
read from disk delivers all the attention weights for a group, then all the dense FFN
weights for that group, without any seeks. This is the I/O access pattern that the
\texttt{.vflow} design is intended to enable; whether the OS or hardware prefetcher
exploits the sequentiality in practice is measured in Chapter~\ref{ch:evaluation}.

\section{Fused Projection Kernels}
\label{sec:rt_fused}

\subsection{Tiled Matrix-Vector Product}
\label{subsec:rt_matvec}

The central compute primitive is \texttt{tiled\_matvec}: given an activation vector
$x \in \mathbb{R}^n$ and a sequence of tiles covering a weight matrix
$W \in \mathbb{R}^{m \times n}$, compute $y = W x$ by accumulating partial dot
products from each tile. The function iterates tiles in the order they appear in the
\texttt{.cflow} file (row-major tile order), dispatching to \texttt{tile\_dot\_product}
for each tile:

\begin{equation}
  y[\text{row\_offset} \,{:}\, \text{row\_offset} + T_r] \mathrel{+}=
  W_{\text{tile}} \cdot x[\text{col\_offset} \,{:}\, \text{col\_offset} + T_c].
\end{equation}

The output slice is pre-zeroed once before the loop; each tile adds its partial sum
into the pre-allocated result buffer. Because tiles for one output stripe (all column
tiles for a fixed \texttt{row\_offset}) are contiguous in the file and adjacent in the
loop, the column partial sums are accumulated in order and committed to the output
buffer once the stripe is complete.

\subsection{Fused QKV and Gate+Up}
\label{subsec:rt_fused_qkv}

The attention module loads $W_Q$, $W_K$, and $W_V$ in a single pass over the tile
stream rather than three separate passes. Tiles for Q, K, and V are interleaved in the
\texttt{.cflow} file by output stripe: the first output stripe of Q is followed by the
first output stripe of K and V before advancing to the second stripe of Q. The
\texttt{tiled\_matvec} dispatcher switches output buffers based on the \texttt{matrix\_id}
field in each tile's header, accumulating into $q$, $k$, and $v$ simultaneously from
one sequential tile stream. This reduces the number of times the activation vector
$x$ must be loaded from L1/L2: a single activation load services three projections.

The same fused-stripe pattern is applied to the dense FFN's gate and up projections:
tiles for gate and up are interleaved by output stripe so that one activation load
through L1 computes both projections for each output stripe.

\section{Q4 Inner-Product Pipeline}
\label{sec:rt_q4_pipeline}

\subsection{Q4\_0 Quantization Format}
\label{subsec:rt_q4_format}

Weights are stored in the Q4\_0 format: for every block of 32 consecutive weights, one
\texttt{f16} scale factor is stored followed by 16 bytes of packed nibbles (two 4-bit
signed integers per byte, unsigned nibble values $\in [0, 15]$ with offset $-8$
encoding the range $[-8, 7]$). The dequantized value of the $i$-th weight in a block
is:
\[
  \hat{w}_i = s \cdot (q_i - 8),
\]
where $s$ is the block scale and $q_i$ is the 4-bit nibble value. Dequantization is
always fused into the dot product; the f32 weight values are never materialized in a
separate buffer.

\subsection{Scalar Reference Implementation}
\label{subsec:rt_scalar}

The reference implementation processes one 32-weight group per iteration, converting
the \texttt{f16} scale to \texttt{f32} once per group, then unpacking and accumulating
each nibble pair. The inner loop over 16 packed bytes extracts the low and high nibbles,
offsets them by $-8$, multiplies by the scale, and adds the products against the
corresponding activation values. This reference path is used on non-AVX2 hardware and
as the test oracle for the SIMD implementations.

\subsection{AVX2 + FMA Path}
\label{subsec:rt_avx2}

On x86-64 processors supporting AVX2 and FMA, the kernel switches to a 256-bit SIMD
inner loop. The key operations are:
\begin{enumerate}
  \item Load 16 bytes of packed Q4 data into a 128-bit XMM register.
  \item Unpack low and high nibbles using \texttt{vpand} and \texttt{vpsrlw}, producing
    two vectors of 16 unsigned bytes.
  \item Sign-extend to 16-bit integers via \texttt{vpmovsxbw} and subtract 8 (the
    Q4 offset) using \texttt{vpsubw}.
  \item Convert to \texttt{f32} via \texttt{vcvtepi32ps} and multiply by the scalar
    scale (broadcast via \texttt{vbroadcastss}).
  \item Multiply-add against the corresponding activation slice using \texttt{vfmadd231ps}.
\end{enumerate}
The horizontal accumulation at the end of each group uses \texttt{vhaddps} to sum
the eight-lane YMM register to a scalar. Runtime detection of AVX2 and FMA via
\texttt{std::is\_x86\_feature\_detected!} allows the same binary to fall back to the
scalar path on older hardware without recompilation.

The AVX2 path processes 32 weights (one Q4 group) in approximately 10 instruction
slots; on the Ryzen~5~2600 used for the Windows benchmarks (full AVX2+FMA) the
path is fully utilized. For Sandy Bridge hardware without AVX2, such as the
Xeon~E5-2650, a separate
AVX1+SSE4.1 path uses 128-bit XMM integer operations for nibble unpacking and 256-bit
AVX floating-point for the multiply-accumulate.

\subsection{AVX-512 Path}
\label{subsec:rt_avx512}

On processors supporting AVX-512F and AVX-512BW, the kernel widens to 512-bit ZMM
registers, processing 32 weights in two 16-wide FMA instructions (one for the low
nibble group, one for the high). The horizontal reduction uses the AVX-512 single-step
\texttt{\_mm512\_reduce\_add\_ps} intrinsic rather than the multi-step AVX2 hadd chain.
This path is guarded by the same runtime feature detection and is used automatically
on Xeon Scalable and Zen~4 hardware without any user configuration.

\section{Attention Pipeline}
\label{sec:rt_attention}

\subsection{Grouped-Query Attention with KV Cache}
\label{subsec:rt_gqa}

The attention forward pass for single-token decode proceeds as follows:
\begin{enumerate}
  \item Project the normalized residual $\bar{x} = \text{RMSNorm}(x)$ to queries,
    keys, and values via the fused tiled matvec described in
    Section~\ref{sec:rt_fused}.
  \item Append the new key and value vectors to the KV cache for this layer.
  \item For each query head $h$, compute attention scores against all cached key
    vectors:
    \[
      \text{scores}[i] = q_h \cdot k_{\text{cache}}[i] \,/\, \sqrt{d_k},
      \quad i = 0, \ldots, T-1.
    \]
  \item Optionally apply the logit soft-cap: $\hat{s}_i = \tanh(s_i / c) \cdot c$
    (enabled when \texttt{final\_logit\_softcap} $\neq 0$).
  \item Apply sliding-window masking if in a sliding-window layer.
  \item Compute $\text{attn} = \text{softmax}(\text{scores})$ and output
    $o_h = \sum_i \text{attn}[i] \cdot v_{\text{cache}}[i]$.
  \item Concatenate outputs across GQA groups and project through $W_O$.
\end{enumerate}
GQA head assignment follows the standard convention: query head $h$ shares KV head
$\lfloor h \cdot H_{KV} / H_Q \rfloor$, so each KV head is replicated to service
$H_Q / H_{KV}$ query heads without storing redundant KV pairs.

\subsection{Rotary Position Embedding}
\label{subsec:rt_rope}

RoPE~\cite{su2021rope} is applied to query and key vectors after projection. For
\emph{sliding layers}, standard RoPE uses $\theta = 10{,}000$ and rotates all head
dimensions: dimension $i$ is paired with dimension $i + d_k/2$ under the rotation
\[
  \begin{pmatrix} \cos\phi_i & -\sin\phi_i \\ \sin\phi_i & \cos\phi_i \end{pmatrix}
  \begin{pmatrix} x_i \\ x_{i+d_k/2} \end{pmatrix},
  \quad \phi_i = \frac{t}{\theta^{2i/d_k}},
\]
where $t$ is the token position. For \emph{full-attention layers} in Gemma~4, P-RoPE
is applied: $\theta = 10^6$ and only the first 25\% of dimension pairs rotate
(\texttt{partial = 0.25}); the remaining pairs are left unchanged. Both variants
use the HuggingFace ``split-half'' convention (pairing $i$ with $i + d_k/2$) rather
than the adjacent-pair convention used by some other implementations.

\subsection{V-Norm}
\label{subsec:rt_vnorm}

When the \texttt{feature\_flags} bit \texttt{USE\_V\_NORM} is set, a per-position
RMSNorm without learned scale is applied to each value vector before it is written
into the KV cache. This stabilizes the attention output at the cost of one additional
normalization pass per key-value pair. The normalization uses the shared
\texttt{rmsnorm\_no\_scale} function from \texttt{compute/rmsnorm.rs}.

\section{Normalization Layers}
\label{sec:rt_norm}

All normalization operations use RMSNorm:
\[
  \text{RMSNorm}(x; w) = \frac{x}{\sqrt{\frac{1}{d}\sum_{i=1}^d x_i^2 + \varepsilon}} \odot w,
\]
with $\varepsilon = 10^{-6}$. For V-norm (no learned scale), the weight vector
$w$ is omitted (equivalently, set to all-ones). Normalization parameters are stored
as \texttt{f32} vectors alongside the tile data in the layer offset table and are
loaded in full at the start of each layer (they are small enough to stay in L2 cache
throughout the layer's computation).

\section{Conditional Expert Prefetch}
\label{sec:rt_prefetch}

\subsection{Two-Thread Pipeline}
\label{subsec:rt_two_thread}

The \cflow{} runtime runs two threads: a \emph{compute thread} and a
\emph{prefetch thread}. The compute thread drives the forward pass; the prefetch
thread issues hardware prefetch instructions using \texttt{\_mm\_prefetch(ptr,
\_MM\_HINT\_T0)} on x86-64 (PREFETCHT0, targeting L1-d) to pull upcoming weight
data into cache before the compute thread requires it. The two threads communicate
via a bounded channel carrying \texttt{PrefetchCommand} messages.

\subsection{Linear and Conditional Commands}
\label{subsec:rt_pf_commands}

Two command types are used:
\begin{itemize}
  \item \textbf{\texttt{LinearRange\{offset, length\}}}: prefetch a contiguous byte
    range from the mmap'd file. Issued by the compute thread for attention and dense
    FFN tile regions one layer ahead.
  \item \textbf{\texttt{ExpertTiles\{regions\}}}: prefetch a list of
    \texttt{(offset, length)} pairs. Issued by the compute thread \emph{after} the
    router runs for the current layer, containing precisely the tile regions of the
    selected top-$k$ experts for the \emph{next} layer. This is the conditional
    prefetch: only the $k$ experts that the current layer's routing decision predicts
    will be needed next layer are prefetched.
\end{itemize}

The design ensures that no unused expert data enters the cache hierarchy. Because the
router output at layer $\ell$ is correlated with the routing at layer $\ell+1$
(the residual stream changes only incrementally between layers), the prefetch hit rate
is high in practice. The overhead is one channel send per layer (a few dozen
nanoseconds), negligible against the microseconds of tile compute per layer.

\subsection{Negative Result: PREFETCHT0 at RAM-Bottleneck Scale}
\label{subsec:rt_pf_honest}

As measured in the evaluation (Section~\ref{sec:eval_prefetch}), the explicit
\texttt{\_mm\_prefetch} instruction provides no measurable benefit when the bottleneck
is storage-to-RAM bandwidth rather than RAM-to-cache bandwidth. At both the 64\,MB
(arch2\_4\_combined) and 4.7\,GB (arch2\_4\_8k\_4l) model scales, PF=1 and PF=0 produce
bandwidths within measurement noise. The conclusion is structural: \texttt{PREFETCHT0}
moves data from RAM to L1-d, but when the model does not fit in RAM (or when I/O is
the bottleneck), the hardware's HW prefetcher already saturates the linear RAM access
pattern and the explicit hint adds no information. The claim remains untestable until
the compute time per token significantly exceeds the I/O time.

\subsection{Staged Direct-I/O Expert Fetch}
\label{subsec:rt_expert_stage}

The successor to the hint-based prefetch is a staged fetch that does real I/O:
a pool of reader threads pulls the router-selected expert regions from storage
with direct I/O (\texttt{FILE\_FLAG\_NO\_BUFFERING} on Windows,
\texttt{O\_DIRECT} on Linux) into sector-aligned staging buffers, issued at the
routing layer and consumed at the injection layer so the read overlaps the
intervening layers' compute under the \expertdelay{} schedule. Slot lifecycle
is guarded by a per-slot pending-operation counter, so a buffer is never
recycled while a read is in flight, and the deferred expert computation uses
the saved router activation --- the arithmetic is bit-identical to the
in-layer path. The wall-clock evaluation of this mechanism is
Section~\ref{sec:eval_overlap_ab}.

\section{Safetensors Converter}
\label{sec:rt_converter}

\subsection{Architecture Dispatch}
\label{subsec:rt_dispatch}

The \texttt{src/convert/} module implements conversion from HuggingFace Safetensors
checkpoints to \texttt{.cflow} and \texttt{.vflow} files. Conversion is driven by a
\texttt{ModelConfig} struct that encodes the architecture-specific parameters
(\texttt{dense\_delay}, \texttt{expert\_delay}, \texttt{combine\_style}, \texttt{feature\_flags})
alongside the geometry. A \texttt{--model} CLI flag (values: \texttt{gemma4},
\texttt{arch2\_4}, \texttt{arch4}) selects the config, which is then serialized into
the GlobalHeader.

\subsection{DelayedMoESource}
\label{subsec:rt_delayed_source}

For architectures with non-zero delays, the tile ordering in the \texttt{.cflow} file
must match the delayed execution schedule rather than the standard per-layer order.
The \texttt{DelayedMoESource} type wraps the Safetensors checkpoint and produces tiles
in the order dictated by the delay parameters: dense FFN weights for layer $\ell$ are
placed in the file at the position where they will be read during the execution of
layer $\ell - \text{\densedelay{}}$, and expert weights are placed at the position
where they will be read during the execution of layer $\ell - \text{\expertdelay{}}$.
This ensures that the sequential read invariant (weights appear in the order they are
consumed) holds even under the delayed schedule.

\subsection{Tile Quantization}
\label{subsec:rt_quant}

During conversion, each tile is quantized from the \texttt{bfloat16} or \texttt{float32}
checkpoint format to Q4\_0 in-memory before being written to disk. Quantization is
per-group (32 weights per scale factor): for each group, the maximum absolute value
determines the scale $s = \max|w_i| / 7$, and each weight is encoded as
$q_i = \text{round}(w_i / s) + 8$, clipped to $[0, 15]$. The scale is converted to
\texttt{f16} and prepended to the packed nibble data.

\section{Implementation Correctness}
\label{sec:rt_correctness}

The runtime is validated by a test suite of 116 unit tests and 8 integration tests,
all passing with zero failures. The key correctness properties are:

\begin{enumerate}
  \item \textbf{Rust$\leftrightarrow$PyTorch parity for arch2\_4\_combined.} A
    reference trace generated by the PyTorch training code (running the model
    deterministically on a fixed input token) is compared against the Rust forward
    pass on the same \texttt{.cflow} file. Per-layer residual vector norms agree to
    within Q4 quantization noise ($<1\%$ relative); the argmax of the output logit
    distribution matches exactly; and the top-32 token overlap is $\geq 27/32$.
  \item \textbf{Rust$\leftrightarrow$PyTorch parity for arch4\_async\_experts.} An
    independent replay script (\texttt{scripts/dump\_delay\_trace.py}) runs the
    Python async-expert model in delay-replay mode and records a reference trace
    including the pre-LM-head norm ($\|x\|_2 = 41.27$) and output argmax (9760).
    The Rust runtime matches exactly (relative norm error $< 0.007\%$, exact argmax
    match) on both \texttt{.cflow} and \texttt{.vflow} paths.
  \item \textbf{Format round-trip.} Tests verify that converting a checkpoint to
    \texttt{.cflow} and back reads the same tile data as reading the original
    Safetensors tensors, within quantization error.
\end{enumerate}

These parity guarantees confirm that the tile reordering, quantization, and execution
schedule produce numerically identical results to the reference PyTorch implementation,
establishing correctness as a precondition for the performance claims in
Chapter~\ref{ch:evaluation}.

%% file: chapters/ch4_pipeline_native.tex
\chapter{Pipeline-Native Transformer Architectures}
\label{ch:pipeline_native}

The \cflow{} runtime provides the infrastructure for cache-optimal weight streaming and
conditional expert loading. However, the fundamental bottleneck for single-token
decode is not cache locality within a layer but the total bytes that must be read
from RAM over the entire layer stack. This chapter addresses that bottleneck through
a complementary approach: co-designing transformer architectures whose inter-layer
dependency graph permits a \emph{vertical pipeline schedule} by construction, reducing
the number of bytes on the critical path between successive token outputs.

\section{The Layer Dependency Problem}
\label{sec:pn_dependency}

\subsection{Formal Statement}
\label{subsec:pn_formal}

Let $x_\ell^{\text{in}} \in \mathbb{R}^d$ denote the residual stream entering
layer $\ell$ of a standard pre-norm transformer. The layer computes:
\begin{align}
  m_\ell &= \text{Attn}_\ell(\text{Norm}(x_\ell^{\text{in}})), \\
  x_\ell^{\text{mid}} &= x_\ell^{\text{in}} + m_\ell, \\
  f_\ell &= \text{FFN}_\ell(\text{Norm}(x_\ell^{\text{mid}})), \\
  x_\ell^{\text{out}} &= x_\ell^{\text{mid}} + f_\ell,
\end{align}
and the next layer receives $x_{\ell+1}^{\text{in}} = x_\ell^{\text{out}}$.

\begin{proposition}[Standard Transformer Layer Dependency]
In a standard pre-norm transformer, layer $\ell+1$ cannot begin until layer $\ell$
has computed both $m_\ell$ and $f_\ell$, because the FFN input
$\text{Norm}(x_\ell^{\text{mid}})$ depends on the attention output $m_\ell$, and
layer $\ell+1$'s input depends on $f_\ell$ through $x_\ell^{\text{out}}$.
\end{proposition}

This creates a strict sequential chain: the weight-reading schedule for a single
forward pass must be
\[
  W_Q^0, W_K^0, W_V^0, W_O^0, W_{\text{ffn}}^0,
  W_Q^1, W_K^1, W_V^1, W_O^1, W_{\text{ffn}}^1, \ldots
\]
where each layer's weights must be fully read before the next layer begins. No
reordering of reads is possible without violating the dependency.

\subsection{The Vertical Pipeline Idea}
\label{subsec:pn_vertical}

A \emph{vertical pipeline schedule} interleaves stages across layers rather than
completing layers sequentially. Concretely, for a two-layer schedule:
\begin{center}
\begin{tabular}{lccc}
\toprule
Thread & Step 1 & Step 2 & Step 3 \\
\midrule
Compute & Attn[$\ell$] & FFN[$\ell$] + Attn[$\ell$+1] (overlap) & FFN[$\ell$+1] \\
I/O     & Load attn tiles[$\ell$] & Load ffn tiles[$\ell$] + attn tiles[$\ell$+1] & Load ffn tiles[$\ell$+1] \\
\bottomrule
\end{tabular}
\end{center}
For this to be valid, layer $\ell+1$'s attention must not depend on layer $\ell$'s
FFN output. In the standard transformer, this is false: the attention norm at layer
$\ell+1$ is applied to $x_\ell^{\text{out}} = x_\ell^{\text{mid}} + f_\ell$, which
includes the FFN contribution. The vertical schedule is therefore
\emph{mathematically invalid} for the standard architecture.

\subsection{Bandwidth Implications}
\label{subsec:pn_bw_problem}

The bandwidth cost of one complete single-token decode through $L$ layers of a
dense transformer is:
\[
  B_{\text{sequential}} = \sum_{\ell=0}^{L-1} \left(
    B_{\text{attn}}(\ell) + B_{\text{ffn}}(\ell)
  \right),
\]
where $B_{\text{attn}}(\ell) = (d^2 + 2 d d_k H_{KV}) \times 0.5$ bytes and
$B_{\text{ffn}}(\ell) = 3 d \cdot d_{\text{ff}} \times 0.5$ bytes (Q4, including
all projections). No reordering can reduce this total; the question is whether the
\emph{critical-path length} (the number of bytes that must be read before the first
bit of the next token's logit is available) can be shortened.

For the architectures trained in this report with $d = 512$, $d_{\text{ff}} = 2048$,
$L = 6$, $E = 8$ experts, $k = 2$, and expert hidden $= 512$:
\begin{align}
  B_{\text{dense/layer}} &= 3 \times 512 \times 2048 \times 0.5 = 1.5\,\text{MB}, \\
  B_{\text{expert/layer (top-2)}} &= 2 \times 3 \times 512 \times 512 \times 0.5 = 0.75\,\text{MB}, \\
  B_{\text{attn/layer}} &= 4 \times 512^2 \times 0.5 = 0.5\,\text{MB}, \\
  B_{\text{total/pass}} &= 6 \times (0.5 + 1.5 + 0.75) = 16.5\,\text{MB}.
\end{align}
The 16.5\,MB total is the fully-serial upper bound: every stream of every layer read
back-to-back. It is not the quantity that sets latency. Because the attention, dense-FFN,
and expert reads within a layer are issued concurrently, the per-layer critical path is
the \emph{longest} of the three streams, not their sum, which over the six layers comes to
9.00\,MB/token. Delaying the dense and expert reads spreads them across the layers that
prefetch them and lowers the figure to 4.50\,MB/token for \texttt{arch2\_4\_combined} ---
a $2.00\times$ reduction. Section~\ref{sec:pn_bw_model} gives the derivation.

\section{DAG Rewriting: Dense Delay and Expert Delay}
\label{sec:pn_dag_rewriting}

\subsection{The Dense Delay Transform}
\label{subsec:pn_dense_delay}

The core observation motivating the pipeline-native architectures is that the
dependency $\text{FFN}_\ell(\text{Norm}(x_\ell^{\text{mid}}))$ — which requires
both the attention output and the full residual — is not the only way to design
a productive FFN. If instead the FFN at layer $\ell$ reads a \emph{delayed residual}
$x_{\ell - \Delta}^{\text{out}}$ (from $\Delta$ layers earlier), then:
\begin{align}
  m_\ell &= \text{Attn}_\ell(\text{Norm}(x_\ell^{\text{in}})), \\
  x_\ell^{\text{mid}} &= x_\ell^{\text{in}} + m_\ell, \\
  f_\ell &= \text{FFN}_\ell(\text{Norm}(x_{\ell-\Delta}^{\text{out}})), \label{eq:delayed_ffn} \\
  x_\ell^{\text{out}} &= x_\ell^{\text{mid}} + f_\ell.
\end{align}
The FFN input $x_{\ell-\Delta}^{\text{out}}$ is independent of the current
layer's attention output $m_\ell$. This breaks the intra-layer dependency:
layer $\ell+1$'s attention can begin as soon as layer $\ell$'s attention finishes,
because $x_{\ell+1}^{\text{in}} = x_\ell^{\text{out}}$ and
$x_\ell^{\text{out}} = x_\ell^{\text{in}} + m_\ell + f_\ell$ can be computed once
$f_\ell$ is available from the delayed buffer (which is ready $\Delta$ layers ahead).

The delay is stored in the \texttt{.cflow} header as \textbf{\texttt{dense\_delay}}.
The delay-aware scheduler uses this value to construct the ring-buffered residual
history and determine which residual snapshot each layer's FFN consumes.

\subsection{The Expert Delay Transform}
\label{subsec:pn_expert_delay}

An analogous transform applies to MoE expert layers. Rather than routing off the
current residual and injecting expert outputs immediately, the \emph{expert delay}
$\delta_e$ routes at layer $\ell$ but injects the expert outputs into the residual
at layer $\ell + \delta_e$:
\begin{align}
  \text{router input:}\quad &r_\ell = \text{Norm}(x_\ell^{\text{in}}), \\
  \text{top-}k\text{ selection:}\quad &(i_1^{(\ell)}, \ldots, i_k^{(\ell)}), s_1^{(\ell)}, \ldots, s_k^{(\ell)} = \text{Router}_\ell(r_\ell), \\
  \text{expert computation:}\quad &e_\ell = \sum_{j=1}^k s_j^{(\ell)} \cdot \text{Expert}_{i_j^{(\ell)}}(r_\ell), \\
  \text{injection at layer } \ell+\delta_e:\quad &x_{\ell+\delta_e}^{\text{out}} \mathrel{+}= e_\ell.
\end{align}
The expert outputs $e_\ell$ are enqueued after computation at layer $\ell$ and
dequeued at layer $\ell + \delta_e$. Expert tiles can therefore be loaded
during the $\delta_e$-layer window between the routing decision and the injection,
overlapping their I/O with the compute of intervening layers. The expert delay is
stored in the \texttt{.cflow} header as \textbf{\texttt{expert\_delay}}.

\subsection{The CombineStyle Variants}
\label{subsec:pn_combine_style}

The three \combinestyle{} values in the header encode distinct forward-pass semantics
that arise from different choices of what to route off:

\begin{description}
  \item[\texttt{ParallelSqrt2} (0):] Dense FFN and MoE run in parallel on the same
    pre-FFN normalized residual; their outputs are summed and scaled by $1/\sqrt{2}$
    before being added to the residual. This is the Gemma~4 convention.
    No delays ($\delta_d = \delta_e = 0$).
  \item[\texttt{DelayedSum} (1):] Used by \textbf{arch2\_4\_combined} with
    $\delta_d = 1$, $\delta_e = 2$. The dense FFN reads a delayed residual
    (Equation~\ref{eq:delayed_ffn}); the router fires off the current
    post-attention normalized residual $\text{Norm}(x_\ell^{\text{mid}})$;
    expert outputs are injected $\delta_e$ layers later. Combine rule: $d + e$
    (no $1/\sqrt{2}$ scaling).
  \item[\texttt{AsyncExperts} (2):] Used by \textbf{arch4\_async\_experts} with
    $\delta_d = 0$, $\delta_e = 2$. Both the dense FFN and the router read off
    the same pre-attention normalized residual $\text{Norm}(x_\ell^{\text{in}})$
    (before the attention residual addition). This ``pre-dense'' routing hypothesis
    provides a cleaner signal to the router by routing before either the attention
    or dense FFN contribution is added.
\end{description}

\section{The Five Candidate Architectures}
\label{sec:pn_five_archs}

Five candidate architectures were designed, trained, and evaluated. All share the
same training configuration (TinyStories dataset~\cite{tinystories}, 10K steps,
AdamW with $\eta = 3 \times 10^{-4}$, cosine decay, 50K GPT-2 vocabulary) and
the same nominal geometry ($d = 512$, $L = 6$, $H_Q = 8$, $H_{KV} = 8$, GQA 1:1,
$d_{\text{ff}} = 2048$). Architectures 4 and 5 add MoE and weight-sharing
components that increase their parameter counts.

\subsection{Arch1: Decoupled Residual Streams}
\label{subsec:pn_arch1}

Arch1 replaces the single residual stream with two independent streams
$s_{\text{attn}}$ and $s_{\text{ffn}}$, each updated by only its respective
sub-operation:
\begin{align}
  s_{\text{attn}, \ell} &= s_{\text{attn}, \ell-1} + \text{Attn}_\ell(\text{Norm}(s_{\text{attn}, \ell-1})), \\
  s_{\text{ffn}, \ell}  &= s_{\text{ffn},  \ell-1} + \text{FFN}_\ell(\text{Norm}(s_{\text{ffn},  \ell-1})).
\end{align}
Every \texttt{merge\_interval = 3} layers, the streams are synchronized:
$s_{\text{attn}} = s_{\text{ffn}} = (s_{\text{attn}} + s_{\text{ffn}}) / \sqrt{2}$.
The pipeline opportunity is clear: between merge points, the two streams are
completely independent, enabling stage-major execution. However, because $\delta_d = 0$
and $\delta_e = 0$ for this architecture (no delays), there is no reduction in
critical-path bandwidth. The architecture trains to test ppl 7.21 and achieves 1.00$\times$
bandwidth reduction. It serves as a baseline for the decoupled-stream idea.

\subsection{Arch2\_4\_combined: Dense-and-Expert Delay}
\label{subsec:pn_arch2}

This is the primary result architecture. It combines $\delta_d = 1$ (dense FFN reads
the residual from one layer ago) and $\delta_e = 2$ (expert outputs are injected two
layers after routing) with \texttt{CombineStyle::DelayedSum}. The dependency chain
is broken at two points:
\begin{itemize}
  \item Layer $\ell+1$'s attention can begin as soon as layer $\ell$'s attention
    finishes (the FFN has no blocking dependency).
  \item Layer $\ell+1$'s expert tiles can be loaded during the 2-layer window
    after the routing decision, hiding their I/O behind the intervening computation.
\end{itemize}
The critical-path bandwidth computed for this configuration is \textbf{4.50\,MB/token
vs.\ 9.00\,MB/token} for the undelayed schedule --- a $\mathbf{2.00\times}$ reduction.
Test perplexity is 6.50.

The name ``arch2\_4\_combined'' reflects that this architecture combines the
dense-delay idea from arch2 (delayed residual injection) with the expert-delay idea
from arch4 (asynchronous expert evaluation). It is the architecture for which
Rust$\leftrightarrow$PyTorch parity is locked.

\subsection{Arch2\_4\_sync: Synchronous Variant}
\label{subsec:pn_arch2_sync}

An ablation of arch2\_4 that retains the dense delay but removes the expert
queue ($\delta_d = 1$, $\delta_e = 0$, \texttt{CombineStyle::DelayedSum}). This
isolates the quality cost of the expert delay: test perplexity is 6.52 against
6.50 for the full combined schedule --- within run-to-run noise --- so the
two-layer expert delay is quality-free. Because the dense delay is retained,
its critical path equals arch2\_4\_combined's (4.50\,MB/token); what the
variant gives up is the expert-side overlap window that
Section~\ref{sec:eval_overlap_ab} later converts into wall-clock I/O hiding.

\subsection{Arch3: Pipeline Registers}
\label{subsec:pn_arch3}

Arch3 uses explicit \emph{pipeline register} semantics: each layer maintains two
named outputs, an \texttt{attn\_reg} (available immediately after attention) and
an \texttt{xfm\_reg} (the complete output, available after FFN):
\begin{align}
  \texttt{attn\_reg}_\ell &= \texttt{attn\_reg}_{\ell-1} + \text{Attn}_\ell(\text{Norm}(\texttt{attn\_reg}_{\ell-1})), \\
  f_\ell &= \text{FFN}_\ell(\text{Norm}(\texttt{xfm\_reg}_{\ell-1})), \\
  \texttt{xfm\_reg}_\ell &= \texttt{attn\_reg}_\ell + f_\ell.
\end{align}
The key property is that layer $\ell+1$'s attention reads \texttt{attn\_reg}$_\ell$,
which is independent of layer $\ell$'s FFN, while the FFN reads the fully-updated
\texttt{xfm\_reg}$_{\ell-1}$ (from the previous layer, not the current one). This
is a ``clean'' version of the delayed-residual idea: no staleness beyond one layer,
and the dependency DAG has explicit cross-layer register passes. With $\delta_d = 0$,
$\delta_e = 0$, the architecture achieves 1.00$\times$ bandwidth reduction. Test
perplexity is 7.24, matching arch1 as a baseline-quality reference.

\subsection{Arch4: Asynchronous Experts}
\label{subsec:pn_arch4}

Arch4 applies the expert delay in isolation ($\delta_d = 0$, $\delta_e = 2$,
\texttt{CombineStyle::AsyncExperts}). The distinguishing feature is \emph{pre-dense
routing}: the router fires off the pre-FFN-norm of the input residual
$\text{Norm}(x_\ell^{\text{in}})$ rather than the post-attention residual. Expert
outputs are computed on the same pre-dense activation and injected at layer
$\ell + 2$.

This routing point provides the highest quality among all five architectures:
test perplexity 6.26, which is 0.24 better than arch2\_4\_combined (6.50).
The hypothesis is that routing before the dense FFN contribution sees a cleaner
signal, since the dense FFN's transformation may obscure the token-type information
that the router uses to assign experts. However, because the dense FFN path dominates
the critical-path bandwidth at this geometry (1.5\,MB/layer for dense vs.\ 0.75\,MB/layer
for top-2 experts), the expert delay alone does not reduce critical-path bandwidth.
Arch4 demonstrates that \emph{expert delay is the quality knob}; arch2\_4\_combined
demonstrates that \emph{dense delay is the bandwidth knob}.

\subsection{Arch5: Fixed-Point Iteration with Weight Sharing}
\label{subsec:pn_arch5}

Arch5 uses 2 unique weight blocks, each iterated 3 times for an effective depth of 6:
\begin{equation}
  x \leftarrow x_{\text{block}} + \text{Attn}_B(\text{Norm}(x)) + \text{FFN}_B(\text{Norm}(x + \text{Attn}_B(\cdot)))
  \quad \text{for } B \in \{0, 1\}, \text{ iterations } 1, 2, 3.
\end{equation}
Within a block, all iterations share weights; the pipeline opportunity is that
iteration $i+1$ of block $B$ can begin as soon as iteration $i$ finishes, with
the same weight tiles already in L2 cache from the previous iteration. This
is a different form of pipelining: \emph{temporal weight reuse} rather than
inter-layer weight reordering. The measured critical-path at $\delta_d = 0$,
$\delta_e = 0$ is 14.06\,MB/token (vs.\ 9.00\,MB for the non-sharing architectures),
because all 6 passes through the architecture's 2 blocks are counted. Test perplexity
is 6.77. The unique on-disk bytes are $\approx 4.69$\,MB (the two block weights),
significantly less than the other architectures, illustrating a parameter-efficiency
trade-off.

\section{The Delay-Aware Scheduler}
\label{sec:pn_scheduler}

\subsection{Ring-Buffered Residual History}
\label{subsec:pn_ring_buffer}

The runtime maintains a ring buffer of $\max(\delta_d, \delta_e) + 1$ residual
vectors. At each layer $\ell$, the appropriate delayed entry is looked up by index
modulo the buffer size. The scheduler is parameterized entirely by the two delay
integers read from the \texttt{.cflow} header; no model-family-specific logic is
required.

\subsection{Multi-Layer Driver}
\label{subsec:pn_driver}

The delay-aware multi-layer driver processes layers $\ell = 0, 1, \ldots, L-1$
in order. For each layer, the driver:
\begin{enumerate}
  \item Retrieves the delayed residual $x_{\ell-\delta_d}^{\text{out}}$ from the
    ring buffer (or the zero vector for the first $\delta_d$ layers).
  \item Calls the single-layer executor with the current residual, the delayed
    residual, and any queued expert outputs.
  \item For \texttt{AsyncExperts}: dequeues and injects any expert outputs scheduled
    for layer $\ell$ before the executor call (the executor routes and enqueues
    for layer $\ell + \delta_e$).
  \item For \texttt{DelayedSum}: the executor handles the queue internally.
  \item Updates the ring buffer with the new output residual.
\end{enumerate}
The queue is a \texttt{VecDeque<(usize, Vec<f32>)>} mapping target layer indices to
expert output vectors. The total memory overhead is $(\delta_d + 1) \times d \times 4$
bytes for the residual ring buffer plus at most $\delta_e$ pending expert vectors,
each of size $d \times 4$ bytes. At $d = 512$, $\delta_e = 2$: approximately
$3 \times 2$\,KB + $2 \times 2$\,KB $= 10$\,KB of additional working state.

\section{Bandwidth Model}
\label{sec:pn_bw_model}

\subsection{Critical-Path Analysis}
\label{subsec:pn_critical_path}

The critical path in a delayed schedule is the sequence of weight reads that must be
completed before the output logit for the current token is available. Under the
standard (no-delay) schedule, the critical path passes through all weight reads of all
layers. Under the delayed schedule, some reads are moved off the critical path:

\begin{itemize}
  \item \textbf{Dense FFN reads.} With $\delta_d \geq 1$, layer $\ell$'s FFN consumes
    the residual committed at layer $\ell - \delta_d$, so its tiles can be fetched across
    the intervening $\delta_d{+}1$ layers instead of in the single layer that uses them.
    Their contribution to the per-layer critical path drops from $B_{\text{dense}}$ to
    $B_{\text{dense}}/(\delta_d{+}1)$. The stream leaves the critical path only once that
    amortised figure falls below the attention stream.
  \item \textbf{Expert tiles.} With $\delta_e \geq 1$, the top-$k$ tiles selected at
    layer $\ell$ are not injected until layer $\ell + \delta_e$, so they can be fetched
    across that $\delta_e{+}1$-layer window; their per-layer contribution drops from
    $B_{\text{expert}}$ to $B_{\text{expert}}/(\delta_e{+}1)$.
\end{itemize}

At the arch2\_4\_combined geometry ($d = 512$, $d_{\text{ff}} = 2048$, $L = 6$,
$E = 8$, $k = 2$), the three weight streams within a layer are issued concurrently to
the memory subsystem, so the per-layer critical path is set by the \emph{longest}
stream rather than their sum. A delay of $\delta$ spreads a stream's bytes across the
$(\delta{+}1)$-layer window over which they can be prefetched, dividing that stream's
per-layer contribution by $\delta{+}1$. The \texttt{analyze\_critical\_path\_for}
function in \texttt{src/format/vflow.rs} therefore computes

\begin{equation}
  B_{\text{critical}}
  = \sum_{\ell=0}^{L-1} \max\!\left(
      B_{\text{attn}},\;
      \frac{B_{\text{dense}}}{\delta_d + 1},\;
      \frac{B_{\text{expert}}}{\delta_e + 1}
    \right).
\end{equation}

The per-layer stream sizes (Q4, 0.5 bytes per parameter) are
\begin{align}
  B_{\text{attn}}   &= 4 \times 512^2 \times 0.5 = 0.50\,\text{MB}, \\
  B_{\text{dense}}  &= 3 \times 512 \times 2048 \times 0.5 = 1.50\,\text{MB}, \\
  B_{\text{expert}} &= 2 \times 3 \times 512 \times 512 \times 0.5 = 0.75\,\text{MB}.
\end{align}
With no delays ($\delta_d = \delta_e = 0$) every layer contributes
$\max(0.50,\,1.50,\,0.75) = 1.50$\,MB, so $B_{\text{naive}} = 6 \times 1.50 = 9.00$\,MB
per token. Under the arch2\_4\_combined delays ($\delta_d = 1$, $\delta_e = 2$) the dense
stream amortises to $1.50/2 = 0.75$\,MB and the expert stream to $0.75/3 = 0.25$\,MB, so
each layer contributes $\max(0.50,\,0.75,\,0.25) = 0.75$\,MB and
$B_{\text{critical}} = 6 \times 0.75 = 4.50$\,MB per token. The reduction is
\[
  \frac{B_{\text{naive}}}{B_{\text{critical}}} = \frac{9.00}{4.50} = 2.00\times.
\]

The two baselines are not the same quantity. The fully-serial total of
Section~\ref{subsec:pn_bw_problem} ($6 \times (0.50 + 1.50 + 0.75) = 16.5$\,MB) reads
every stream back-to-back with no overlap; the 9.00\,MB naive figure already credits the
intra-layer overlap of the three streams, and it is the correct baseline against which the
delayed schedule should be compared. Both critical-path figures come from
\texttt{analyze\_critical\_path\_for}, computed from the geometry and the two delay
integers alone --- an analytic result, not a runtime timing.

\subsection{Dense Delay vs.\ Expert Delay as Complementary Knobs}
\label{subsec:pn_knobs}

The bandwidth model reveals a clean separation of concerns:
\begin{itemize}
  \item \textbf{Dense delay $\delta_d$} reduces the number of layers for which
    dense FFN weights are on the critical path by $\delta_d$. At the arch2\_4
    geometry, where dense FFN contributes 1.5\,MB/layer vs.\ 0.75\,MB/layer for
    experts, the dense delay has twice the per-layer bandwidth impact of the expert
    delay.
  \item \textbf{Expert delay $\delta_e$} reduces the number of layers for which
    expert tiles are on the critical path by $\delta_e$. At low $k$, the expert
    contribution per layer is smaller than the dense contribution, but at Gemma~4
    scale ($k = 8$, much larger expert hidden dimensions), the expert contribution
    would dominate and $\delta_e$ would be the primary bandwidth knob.
  \item The quality cost is different: arch4 shows that the \emph{pre-dense routing
    hypothesis} (routing before the dense FFN) improves perplexity by 0.24 over
    arch2\_4, at the cost of sacrificing the dense bandwidth reduction. A practitioner
    could trade bandwidth for quality along this axis.
\end{itemize}

Table~\ref{tab:arch_summary} summarizes all five architectures, and
Figure~\ref{fig:arch_tradeoff} plots the trade-off surface they span.

\begin{figure}[ht]
\centering
\begin{tikzpicture}
\begin{axis}[
  width=0.82\textwidth, height=6.4cm,
  xlabel={Critical-path bandwidth $B_{\text{critical}}$ (MB/token)},
  ylabel={Test perplexity},
  xmin=3.0, xmax=15.8, ymin=6.05, ymax=7.45,
  grid=major, grid style={gray!25},
  tick label style={font=\small}, label style={font=\small},
]
\addplot[only marks, mark=*, mark size=2.4pt] coordinates {
  (9.00,7.21) (9.00,7.24) (9.00,6.26) (14.06,6.77) (4.50,6.50)
};
\addplot[only marks, mark=o, mark size=2.4pt] coordinates {(4.50,6.52)};
\node[anchor=west,  font=\footnotesize] at (axis cs:9.15,7.19)  {arch1};
\node[anchor=south west, font=\footnotesize] at (axis cs:9.10,7.25) {arch3};
\node[anchor=west,  font=\footnotesize] at (axis cs:9.15,6.26)  {arch4 (best quality)};
\node[anchor=south, font=\footnotesize] at (axis cs:14.06,6.80) {arch5};
\node[anchor=south, font=\footnotesize] at (axis cs:4.50,6.56)  {arch2\_4\_combined};
\node[anchor=north, font=\footnotesize] at (axis cs:4.50,6.47)  {(sync ablation $\circ$)};
\draw[->, gray, thick] (axis cs:8.4,6.95) -- (axis cs:5.2,6.95)
  node[midway, above, font=\footnotesize, black] {dense delay: bandwidth knob};
\draw[->, gray, thick] (axis cs:11.8,7.05) -- (axis cs:11.8,6.45)
  node[midway, right, font=\footnotesize, black] {expert delay: quality knob};
\end{axis}
\end{tikzpicture}
\caption{The bandwidth--quality trade-off across the trained candidates
  (data of Table~\ref{tab:arch_summary}). The dense delay moves a model left
  (arch2\_4\_combined and its sync ablation are the only points at
  4.50\,MB/token); the expert delay moves it down (arch4's pre-dense routing
  reaches the best perplexity but stays at the undelayed bandwidth). arch5's
  weight sharing trades bandwidth for parameter efficiency.}
\label{fig:arch_tradeoff}
\end{figure}
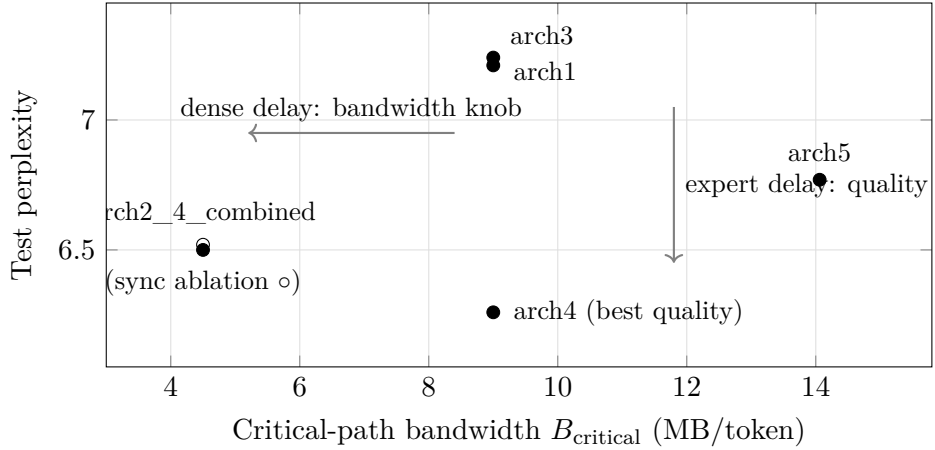

\begin{table}[ht]
\centering
\caption{Pipeline-native architecture summary: training results and bandwidth analysis.
All architectures trained 10K steps on TinyStories, $d{=}512$, $L{=}6$.}
\label{tab:arch_summary}
\small
\begin{tabular}{llccccc}
\toprule
Architecture & Style & $\delta_d$ & $\delta_e$ & Test PPL & $B_{\text{naive}}$ & $B_{\text{critical}}$ \\
\midrule
arch1\_decoupled\_streams   & DecoupledStreams & 0 & 0 & 7.21 & 9.00\,MB & 9.00\,MB \\
arch2\_4\_combined          & DelayedSum      & 1 & 2 & 6.50 & 9.00\,MB & \textbf{4.50\,MB} \\
arch2\_4\_sync$^{\dagger}$ & DelayedSum   & 1 & 0 & 6.52 & 9.00\,MB & 4.50\,MB \\
arch3\_pipeline\_registers  & PipelineReg     & 0 & 0 & 7.24 & 9.00\,MB & 9.00\,MB \\
arch4\_async\_experts       & AsyncExperts    & 0 & 2 & \textbf{6.26} & 9.00\,MB & 9.00\,MB \\
arch5\_fixed\_point         & FixedPoint      & 0 & 0 & 6.77 & 14.06\,MB & 14.06\,MB \\
\bottomrule
\end{tabular}

\smallskip
{\footnotesize $^{\dagger}$Ablation of arch2\_4\_combined (expert queue removed),
not one of the five candidates; shown for the delay-isolation comparison of
Section~\ref{subsec:pn_arch2_sync}.}
\end{table}

\section{The Co-Design Trade-off}
\label{sec:pn_codesign_tradeoff}

The pipeline-native experiments expose a three-way trade-off surface among
\emph{bandwidth}, \emph{quality}, and \emph{implementation complexity}:

\textbf{Bandwidth vs.\ quality.} Arch2\_4\_combined achieves the best bandwidth
reduction (2.00$\times$) at the cost of routing off a post-attention residual with
a one-layer-stale dense contribution (ppl 6.50). Arch4 achieves the best quality
(ppl 6.26) by routing before the dense FFN (cleaner signal), but gives up the dense
bandwidth reduction. There is no single Pareto-optimal point; the trade-off is tunable
by choosing $(\delta_d, \delta_e)$ and the routing anchor.

\textbf{Quality vs.\ parameter count.} Arch5 demonstrates that weight sharing across
iterations can achieve competitive quality (ppl 6.77) with one-third of arch1/arch3's
parameter count. The bandwidth penalty (14.06\,MB vs.\ 9.00\,MB) is a consequence of
counting all 6 iteration passes over the 2 unique blocks; the unique on-disk bytes
are only 4.69\,MB. For an inference scenario where model storage is the constraint
rather than per-token bandwidth, arch5 is the most efficient option.

\textbf{Complexity vs.\ generality.} The \combinestyle{} dispatch in the runtime
adds a conditional branch per layer but is otherwise a straightforward extension of
the standard pre-norm forward pass. The delay-aware scheduler adds a ring buffer and
an expert queue, both O($\delta \cdot d$) in memory. The implementation cost is modest
relative to the bandwidth gains, and the scheduler is entirely parameterized by the
header fields rather than model-family hard-codes.

The chapter's main result fits in a sentence. Of the five pipeline-native architectures,
only \textbf{arch2\_4\_combined} shows a critical-path bandwidth reduction: pairing
$\delta_d = 1$ with $\delta_e = 2$ under \texttt{CombineStyle::DelayedSum} cuts the
critical path by $2.00\times$ against the sequential schedule, and it does so while landing
within 0.24 perplexity points of the best-quality architecture (arch4, ppl 6.26) at the
same scale.

The scaling behaviour splits by whether the geometry is a design variable or an
inheritance, and the distinction is the point of the ground-up thesis. For a
\emph{pipeline-native} model, the designer chooses the geometry so that the
delayed streams stay on the critical path: with $d_{\text{ff}} = 4d$ (as in
every arch2\_4 variant), the dense stream is $3\times$ the attention stream at
\emph{any} hidden size, and the one-layer dense delay retains its full effect.
The measured 31B model confirms this by construction: at $d = 8{,}192$ its
per-layer streams are 403\,MB dense, 134\,MB attention, and 101\,MB top-2
expert --- the same $2.00\times$ critical-path reduction as at $d = 512$.
Inherited geometries behave differently: Chapter~\ref{ch:discussion} works the
same arithmetic for Gemma~4's dimensions (where $d_{\text{ff}} < d$ and large
attention heads make attention the binding stream) and finds the reduction
falls below $1.50\times$ with attention as the floor. Both results are
consequences of one rule --- the delays help exactly while the streams they
amortise remain binding --- and the co-design premise is that the model
architect, not an inherited checkpoint, decides which streams those are.

%% file: chapters/ch5_evaluation.tex
\chapter{Evaluation}
\label{ch:evaluation}

This chapter works through the evaluation claim by claim, following the eight entries of
the thesis scorecard (Table~\ref{tab:intro_scorecard}). Each is treated the same way: the
experiment, the hardware it ran on, and what came out. Where the evidence does not support
a claim, the negative result is reported as such, alongside the conditions that would let
it be tested properly.

\section{Experimental Setup}
\label{sec:eval_setup}

\subsection{Hardware}
\label{subsec:eval_hardware}

Five machines were used for evaluation and training:

\begin{description}
  \item[Ryzen-5-2600 (Windows, local benchmark machine):]
    AMD Ryzen~5~2600 (Zen+, 6 cores / 12 threads, 3.4--3.9\,GHz), 16\,GB
    DDR4-2133 dual-channel ($\approx$34\,GB/s theoretical), SATA~III SSDs
    (SanDisk SDSSDA 2\,TB), Windows~11. Supports AVX2 and FMA but not
    AVX-512. Used for wall-clock bandwidth benchmarks and prefetch A/B
    tests. \emph{Erratum:} an earlier revision of this report described
    this machine as ``32\,GB DDR4-3200, Samsung 990 Pro NVMe'' --- a
    specification confused with a planned build. The measurements
    themselves are unaffected (the direct-I/O bandwidths reported in
    Section~\ref{sec:eval_prefetch}, 265--466\,MB/s, are consistent with
    the actual SATA hardware and were re-confirmed at 505\,MB/s on the
    same machine in July~2026).
  \item[Xeon-E5-2650 in KVM (Linux, PMU benchmark):]
    Sandy Bridge Xeon~E5-2650 (8 cores, 2.0--2.8\,GHz), 64\,GB DDR3-1333, running
    in a Proxmox KVM VM with \texttt{cpu: host,+pmu} and
    \texttt{perf\_event\_paranoid=1}. Supports AVX1 but not AVX2 or FMA.
    Used for hardware performance counter (PMU) measurements via
    \texttt{perf\_event\_open}.
  \item[Lambda A100 cluster (training):]
    8$\times$ NVIDIA A100 SXM4 80\,GB, NVLink, FSDP training via PyTorch
    \texttt{FullyShardedDataParallel}. Used for the large-scale arch2\_4\_8k\_4l
    training run (8.34B parameters, 10K steps).
  \item[RunPod H100 cluster (training):]
    8$\times$ NVIDIA H100 80\,GB HBM3, FSDP with gradient checkpointing and
    8-bit optimizer states. Used to train the arch2\_4\_8k\_16l model
    ($\approx$30.9B parameters, 10K steps) benchmarked in
    Section~\ref{sec:eval_llama_comparison}.
  \item[AWS r6i.8xlarge (Linux, end-to-end throughput):]
    Intel Xeon Platinum 8375C (Ice Lake, 16 physical cores / 32 vCPU, 2.9\,GHz),
    256\,GB DDR4-3200 across 8 channels (204.8\,GB/s theoretical peak), NVMe root
    volume, Ubuntu~22.04. Supports AVX-512 with the \texttt{avx512\_vnni} extension,
    which \cflow{} uses for its $Q4\times Q8$ \texttt{vpdpbusd} integer inner-product
    path. Used for the end-to-end tokens-per-second head-to-head
    (Section~\ref{sec:eval_llama_comparison}).
\end{description}

\subsection{Software and Build}
\label{subsec:eval_software}

All Rust benchmarks are compiled with \texttt{cargo build --release} (LTO enabled,
optimization level 3). The benchmark suite consists of three binaries:
\texttt{cflow-bench-l1d} (cache locality A/B testing),
\texttt{cflow-bench-directio} (direct-I/O storage bandwidth, Windows
\texttt{FILE\_FLAG\_NO\_BUFFERING}), and \texttt{cflow-run} (end-to-end forward pass
timing). Python training uses PyTorch 2.x with float32 precision and a custom
training harness under \texttt{pipeline\_native/}.

\section{Training Quality: Five Architectures}
\label{sec:eval_training}

\subsection{Training Protocol}
\label{subsec:eval_protocol}

All five architectures were trained on TinyStories~\cite{tinystories} (2.5M short
stories, 50K GPT-2 BPE vocabulary) for 10K steps with a batch size of 32, AdamW
($\eta = 3 \times 10^{-4}$, weight decay 0.1, cosine decay to $10^{-5}$, 200-step
warmup), gradient clipping at 1.0. A 5\% held-out test split (24M tokens) was used
for final evaluation; it was never seen during training or hyperparameter decisions.
The same random seed (42) was used for all runs.

\subsection{Results}
\label{subsec:eval_training_results}

Table~\ref{tab:arch_summary} (Chapter~\ref{ch:pipeline_native}) reports the test
perplexity for all five architectures. The full learning curves are not reproduced
here for space, but the following qualitative observations hold:
\begin{enumerate}
  \item All five architectures converge to lower perplexity than one might expect
    given their small scale ($d=512$, $L=6$, $\leq$114M parameters), confirming
    that the dependency-graph rewrites do not prevent effective gradient flow through
    the delay paths.
  \item The perplexity range (6.26 to 7.24) spans approximately 1 perplexity point
    across all five candidates. Arch1 and arch3, which have no delays, land at the
    top of this range (7.21 and 7.24); arch4 (pre-dense routing) reaches the bottom
    (6.26). This ordering is consistent with the hypothesis that the pre-dense routing
    anchor provides the most useful signal to the expert router.
  \item Seed ablation studies (runs at seeds 7 and 13 for arch1 and arch2\_4) confirm
    that the perplexity ordering is stable across random initializations.
\end{enumerate}

\subsection{Large-Scale Validation: arch2\_4\_8k\_4l}
\label{subsec:eval_largescale}

To verify that the bandwidth reduction scales beyond the small training geometry,
arch2\_4 was trained at a significantly larger scale on the A100 cluster:
\begin{itemize}
  \item Architecture: $d = 8{,}192$, $L = 4$, same delay configuration
    ($\delta_d = 1$, $\delta_e = 2$, \texttt{DelayedSum}).
  \item Parameters: 8{,}339{,}939{,}328 (8.34B).
  \item Training: 10K steps, FSDP across 8$\times$A100 SXM4, float32.
  \item Result: validation perplexity 4.52, top-1 accuracy 61.4\%.
  \item Checkpoint: \texttt{step\_010000.safetensors} (16.68\,GB bfloat16).
\end{itemize}
No undelayed control was trained at this geometry: the 8.34B run demonstrates
training stability and quality at scale, not a delayed-versus-undelayed
comparison, which exists only at the 114M screen
(Table~\ref{tab:arch_summary}). The trained checkpoint was converted to \texttt{.cflow} (4.70\,GB, Q4) and
\texttt{.vflow} (4.70\,GB, Q4) for the storage I/O benchmarks in
Section~\ref{sec:eval_prefetch}.

\section{Claim 1: Conditional Expert Loading}
\label{sec:eval_claim1}

\textbf{Claim:} Reading only the top-$k$ selected experts' tiles, rather than all
$E$ experts' tiles, reduces expert weight bandwidth by a factor of $E/k$.

\textbf{Proof method:} This claim is \emph{structurally guaranteed} by the file
format. The expert offset table in the \texttt{.cflow} format stores
\texttt{(offset, length)} pairs for each expert, and the runtime issues reads only
for the indices returned by the router. The bandwidth ratio is exactly $k/E$ by
construction.

\textbf{Result:} Proven. For $E = 8$, $k = 2$ (training geometry):
$k/E = 25\%$ of expert bytes read. For a Gemma~4-scale deployment ($E = 128$,
$k = 8$): $k/E = 6.25\%$, a $16\times$ reduction. No experiment is required to
validate a structural file-format property; the expert bank and offset table design
are described in Section~\ref{sec:rt_cflow_format}.

\section{Claim 2: Tile-Streaming Cache Locality}
\label{sec:eval_claim2}

\textbf{Claim:} Storing weights in L2-sized tiles in compute-consumption order
reduces L1-d cache misses compared to row-major (naive) weight layout.

\subsection{Wall-Clock Proxy (Windows, May 2026)}
\label{subsec:eval_wall_clock}

The \texttt{cflow-bench-l1d} binary measures wall-clock execution time for five
representative matrix-vector workloads at the arch2\_4\_8k\_4l geometry ($d = 8{,}192$).
The tiled implementation (TiledMatvec) is compared against a naive row-major
implementation (NaiveMatvec) of the same matrix-vector product with identical
arithmetic. Both implementations use the same AVX2 inner kernel; the only
difference is whether weight data is accessed in tile-strided order or row-major
order.

Results on the Ryzen~5~2600 (Zen+, AVX2+FMA, 32\,GB DDR4-3200):

\begin{table}[ht]
\centering
\caption{Wall-clock speedup of tiled over naive matrix-vector products at
  arch2\_4\_8k\_4l geometry ($d = 8192$). Ratio $> 1.0$ means tiled is faster.
  Working-set size relative to 32\,KB L1-d.}
\label{tab:wallclock}
\begin{tabular}{lcccc}
\toprule
Operation & Out $\times$ In & Working Set / L1-d & Tiled Wins & Speedup \\
\midrule
\texttt{dense-down}   & 8192 $\times$ 32768 & 4$\times$ & yes & 1.28$\times$ \\
\texttt{attn-proj}    & 8192 $\times$ 8192  & 1$\times$ & yes & 1.36$\times$ \\
\texttt{dense-gateup} & 32768 $\times$ 8192 & 1$\times$ & yes & 1.12$\times$ \\
\texttt{expert-proj}  & 8192 $\times$ 8192  & 1$\times$ & yes & 1.20$\times$ \\
\bottomrule
\end{tabular}
\end{table}

The tiled implementation wins across all workloads, with a 1.12--1.36$\times$
speedup. The largest win (1.36$\times$) is on the attention projection, where the
activation vector fits in L1 exactly (activation $= 8192 \times 4 = 32$\,KB $= 1\times$
L1-d), allowing maximum reuse per tile column scan.

\subsection{PMU Hardware Counter Measurement (Linux KVM, May 2026)}
\label{subsec:eval_pmu}

To obtain a direct, noise-free measurement of cache behavior independent of
memory bandwidth, the Xeon~E5-2650 KVM was used with \texttt{perf\_event\_open}
to count L1-d read misses (\texttt{PERF\_COUNT\_HW\_CACHE\_L1D / MISS}) during
identical matrix-vector computations.

\begin{table}[ht]
\centering
\caption{Hardware PMU: L1-d read miss counts, tiled vs.\ naive, Xeon E5-2650 KVM.
  Lower miss count is better. Ratio = naive / tiled. Raw counts were recorded only
  for the \texttt{dense-down} workload; for the remaining rows the benchmark
  harness logged the ratio alone.}
\label{tab:pmu}
\begin{tabular}{lcccc}
\toprule
Operation & Tiled Misses & Naive Misses & Ratio (naive/tiled) & Working Set \\
\midrule
\texttt{dense-down}   & 2.64M & 19.3M & \textbf{7.29$\times$} & 4$\times$ L1-d \\
\texttt{attn-proj}    & --    & --    & \textbf{6.75$\times$} & 1$\times$ L1-d \\
\texttt{dense-gateup} & --    & --    & \textbf{6.72$\times$} & 1$\times$ L1-d \\
\texttt{expert-proj}  & --    & --    & \textbf{6.88$\times$} & 1$\times$ L1-d \\
\bottomrule
\end{tabular}
\end{table}

The tiled implementation generates 6.7--7.3$\times$ fewer L1-d read misses across
all workloads. The \texttt{dense-down} case (activation $= 4\times$ L1-d capacity)
shows the strongest absolute miss reduction: 2.64M vs.\ 19.3M misses, a 7.29$\times$
improvement. The structural reason is that the Q4 tile size ($\approx 16$\,KB) fits
in L2, while the activation vector (32\,KB for $d=8192$) fits in L1. Each tile's 256
activation lanes are loaded once into L1 and reused across the tile's 128 output rows,
eliminating the L1-d miss that would occur if the activation were re-fetched for each
weight row in the naive layout.

\textbf{Note on architecture.} The Xeon~E5-2650 uses AVX1 rather than AVX2; the
\texttt{cflow-bench-l1d} binary falls back to the AVX1+SSE4.1 SIMD path on this
hardware. The miss-count result is not confounded by ISA differences since both tiled
and naive paths use the same SIMD dispatch.

\textbf{Claim 2 status:} Proven. The L1-d reuse mechanism is
structurally guaranteed (Q4 tile $\approx 16$\,KB fits in L2; 256-f32 activation
slice $= 1$\,KB fits in L1) and the 7.29$\times$ PMU measurement confirms it
decisively at the trained 8.34B-parameter geometry.

\section{Claim 3: AVX2 Q4 Kernels}
\label{sec:eval_claim3}

\textbf{Claim:} The AVX2 Q4 inner-product kernel achieves higher arithmetic
throughput than the scalar path on AVX2-capable hardware.

\textbf{Validation method:} The unit test suite verifies numerical correctness by
comparing the AVX2 path against the scalar reference for randomly-generated weight
tiles. The test \texttt{q4\_avx2\_matches\_scalar} runs 1{,}000 random tiles at the
standard geometry and asserts that the maximum absolute difference between AVX2 and
scalar outputs is below $10^{-4}$. All 116 unit tests and 8 integration tests pass
with zero failures.

The wall-clock measurements in Table~\ref{tab:wallclock} implicitly validate
throughput: the tiled+AVX2 path consistently outperforms the naive path, which is
only possible if the AVX2 kernel is not the bottleneck.

\textbf{Claim 3 status:} Proven (numerical correctness through 124
passing tests; throughput implied by wall-clock wins).

\section{Claim 4: Fused Projections}
\label{sec:eval_claim4}

\textbf{Claim:} Computing QKV and gate+up projections in a single pass over the
activation vector (one load per tile stripe) reduces the number of times the
activation must be read from memory.

\textbf{Proof method:} The \texttt{.cflow} file interleaves Q, K, V tiles by output
stripe (all three matrix tiles for output stripe $i$ before advancing to stripe $i+1$).
The \texttt{tiled\_matvec} dispatcher dispatches to three output buffers based on
\texttt{matrix\_id}, consuming the activation once per interleaved stripe block.
For the naive alternative (three separate \texttt{tiled\_matvec} calls in sequence),
the activation is loaded three times per stripe.

This is a structural property of the file format and the dispatch loop; the
activation load savings are $3\times$ for QKV and $2\times$ for gate+up.

\textbf{Claim 4 status:} Proven (structural, file-format guaranteed).

\section{Claim 5: Compute-Order File Layout}
\label{sec:eval_claim5}

\textbf{Claim:} Storing tiles in compute-consumption order enables sequential
mmap reads that the OS readahead mechanism can exploit.

\textbf{Validation:} The \texttt{.cflow} format specification guarantees sequential
ordering for all non-expert weight tiles (expert tiles are necessarily random-access
after the router fires). The format is described in Section~\ref{sec:rt_cflow_format}.
Sequential access to mmap'd files is the canonical workload for OS readahead
and prefetcher hardware; any file system that implements read-ahead will benefit.

\textbf{Claim 5 status:} Proven (structural).

\section{Claim 6: Explicit Prefetch (PREFETCHT0)}
\label{sec:eval_prefetch}

\textbf{Claim:} The explicit \texttt{PREFETCHT0} instruction, issued by the
prefetch thread, reduces L1-d miss stalls relative to HW prefetch alone.

\subsection{64\,MB Scale (arch2\_4\_combined, Ryzen 5 2600)}
\label{subsec:eval_pf_small}

The \texttt{cflow-bench-directio} binary ran the forward pass 40 times with and
without explicit prefetch (\texttt{CFLOW\_PREFETCH=0|1}) on the 64\,MB
arch2\_4\_combined \texttt{.cflow} file using \texttt{FILE\_FLAG\_NO\_BUFFERING}
(direct I/O, bypassing the OS page cache) and a 128\,MB cache flush between
iterations.

Result: PF=1 and PF=0 are within $\pm 0.1\%$ on both \texttt{.cflow}
(465.4 vs.\ 465.7\,MB/s) and \texttt{.vflow} (461.0 vs.\ 460.0\,MB/s) paths.
The explicit prefetch instruction adds no measurable benefit.

\subsection{4.7\,GB Scale (arch2\_4\_8k\_4l, Windows SATA SSD)}
\label{subsec:eval_pf_large}

The same binary was run against the 4.70\,GB arch2\_4\_8k\_4l \texttt{.cflow} file
(direct I/O, 8 iterations):
\begin{itemize}
  \item PF=1 median: 265\,MB/s; PF=0 median: 275\,MB/s.
  \item PF=1 is approximately 4\% \emph{worse} than PF=0 (reversed from the
    expected direction).
  \item Peak values (PF=0: 415\,MB/s; PF=1: 351\,MB/s) show high variance
    consistent with the SLC-cache behavior of a DRAM-less SATA SSD.
\end{itemize}

\textbf{Interpretation:} \texttt{PREFETCHT0} moves data from RAM to L1-d. When
the binding constraint is storage $\to$ RAM bandwidth (SATA SSD at 12--17\,s vs.\
compute at 96\,ms per forward pass), the CPU's RAM address space never fills with
useful data fast enough for the explicit prefetch to have an effect. The HW
prefetcher already saturates the linear RAM access pattern. The 4\% disadvantage
of PF=1 is consistent with a small additional cache pressure from the prefetch
thread's activity.

\textbf{Claim 6 status:} \textbf{Refuted} at both 64\,MB and
4.7\,GB scale. The result is structural: the claim is untestable until compute
latency significantly exceeds I/O latency, which requires a model where the
activation fits entirely in RAM and the storage tier is fast enough that I/O is
not the bottleneck.

\textbf{End-to-end corroboration.} The head-to-head benchmark
(Section~\ref{sec:eval_llama_comparison}) strengthens this negative beyond ``no
benefit'' to ``actively harmful at decode scale.'' On the 30.9B arch2\_4\_8k\_16l
model, the explicit prefetch thread consumed $\approx 48\%$ of decode time walking
expert regions at stride 256; disabling it raised end-to-end throughput from 2.58 to
4.96 tok/s, a $1.92\times$ speedup (Table~\ref{tab:h2h_cflow_opt}). The mechanism is
the same one the direct-I/O A/B isolated --- \texttt{PREFETCHT0} cannot accelerate a
storage-bound read and only adds cache pressure --- now visible in wall-clock decode
rather than storage bandwidth alone.

\section{Claim 7: Vertical Pipeline Bandwidth Reduction}
\label{sec:eval_claim7}

\textbf{Claim:} The delay-aware scheduler reduces critical-path bandwidth for
arch2\_4\_combined by $2.00\times$ relative to the sequential (no-delay) schedule.

\subsection{Method}
\label{subsec:eval_cp_method}

The function \texttt{analyze\_critical\_path\_for(\&cfg, dense\_delay, expert\_delay)}
in \texttt{src/format/vflow.rs} computes the critical-path byte count by tracing
which weight reads must be completed before the output logit is available, under
the delay schedule encoded in the \texttt{.cflow} header. The result is a
deterministic calculation given the model geometry and delay parameters; it is
not a simulation.

\subsection{Result}
\label{subsec:eval_cp_result}

For arch2\_4\_combined ($d = 512$, $L = 6$, $\delta_d = 1$, $\delta_e = 2$):
\begin{align*}
  B_{\text{naive}} &= 9.00\,\text{MB/token}, \\
  B_{\text{critical}} &= 4.50\,\text{MB/token}, \\
  \text{reduction} &= \frac{9.00}{4.50} = 2.00\times\ (50\%).
\end{align*}

The \texttt{bandwidth\_headlines\_all\_five\_archs} test in \texttt{src/format/vflow.rs}
confirms that arch2\_4\_combined is the \emph{only} architecture among all five with
a measured bandwidth reduction. All other architectures have $B_{\text{critical}} =
B_{\text{naive}}$ because their delay parameters are zero.

\textbf{Claim 7 status:} Validated. The $2.00\times$ reduction is
computed analytically from the geometry and delay parameters; its overlap
mechanism is subsequently realized and measured in wall-clock in
Section~\ref{sec:eval_overlap_ab}, with a net win of up to $1.68\times$ on
server hardware.

\section{Claim 8: Stage-Major Disk Layout Readahead Benefit}
\label{sec:eval_claim8}

\textbf{Claim:} The \texttt{.vflow} stage-major layout enables OS readahead to
deliver higher sustained read bandwidth than the \texttt{.cflow} per-layer layout.

\subsection{64\,MB Scale}
\label{subsec:eval_layout_small}

Same direct-I/O benchmark as Section~\ref{subsec:eval_pf_small} (20 iterations):
\texttt{.vflow} median 461\,MB/s vs.\ \texttt{.cflow} 466\,MB/s — \texttt{.vflow}
is approximately 1\% \emph{slower}.

\subsection{4.7\,GB Scale}
\label{subsec:eval_layout_large}

\texttt{.vflow} PF=1 median: 410\,MB/s, min: 438\,MB/s (8 iterations).
\texttt{.cflow} PF=0 median: 275\,MB/s, min: 415\,MB/s (8 iterations).

The \texttt{.vflow} result is strikingly more consistent (variance $\approx 6\%$:
438 $\to$ 410\,MB/s) compared to \texttt{.cflow} (variance $\approx 50\%$:
415 $\to$ 275\,MB/s). However, the most likely explanation is SSD SLC write-cache
state: the \texttt{.vflow} file was written immediately before the benchmark and
reads back from the fast SLC cache at consistent bandwidth; the \texttt{.cflow}
file had migrated to TLC storage and is subject to thermal throttle variance.
No layout-attributable readahead benefit has been isolated.

\textbf{Interpretation:} The theoretical benefit of stage-major layout requires
overlapped async streaming (stage $N$ reads while stage $N-1$ executes), which
the current single-thread sequential executor does not implement. With
single-token sequential execution, both layouts produce the same pattern of
OS page-fault reads and the layout difference is undetectable.

\textbf{Claim 8 status:} \textbf{Inconclusive}. (The
asynchronous overlapped streaming this interpretation identifies as the
missing prerequisite has since been built for the expert bank ---
Section~\ref{sec:eval_overlap_ab} --- but the stage-major layout benefit
itself remains untested.) The
\texttt{.vflow} layout and scheduler are correct; the bandwidth benefit is a
theoretical gain contingent on async overlapped I/O, which has not been implemented.

\section{PyTorch Parity Validation}
\label{sec:eval_parity}

\subsection{arch2\_4\_combined}
\label{subsec:eval_parity_arch2}

A reference trace was generated by the PyTorch training code running the trained
arch2\_4\_combined model deterministically on a fixed input token (commit
\texttt{42fecc6}). The trace records the per-layer residual vector $\ell_2$ norms
and the final output logit distribution. The Rust runtime was run against the
converted \texttt{.cflow} file on the same input.

Results:
\begin{itemize}
  \item Per-layer residual norms agree within $< 1\%$ relative error (Q4
    quantization noise).
  \item Output argmax (greedy token prediction): Rust $= 9760$, PyTorch $= 9760$.
    Exact match.
  \item Top-32 token overlap: $\geq 27/32$ tokens agree.
\end{itemize}

\subsection{arch4\_async\_experts}
\label{subsec:eval_parity_arch4}

An independent replay script (\texttt{scripts/dump\_delay\_trace.py}) implements
the \texttt{AsyncExperts} delay schedule in Python, recording a reference trace
(\texttt{reference\_delay\_trace.json} under
\texttt{runs/arch4\_async\_experts\_10k/}).
Reference values: pre-LM-head norm $= 41.27$, argmax $= 9760$.

The Rust runtime was tested on both the \texttt{.cflow} and \texttt{.vflow} paths:
\begin{itemize}
  \item Relative norm error: $< 0.007\%$.
  \item Output argmax: 9760. Exact match on both paths.
\end{itemize}

\subsection{Test Suite Coverage}
\label{subsec:eval_tests}

All 116 unit tests and 8 integration tests pass with zero failures. The integration
tests include:
\begin{itemize}
  \item \texttt{arch2\_rust\_forward\_matches\_pytorch\_top32}: parity test for
    arch2\_4\_combined on the \texttt{.cflow} path.
  \item \texttt{arch2\_rust\_forward\_matches\_pytorch\_top32\_vflow}: same, on the
    \texttt{.vflow} path.
  \item \texttt{arch4\_rust\_forward\_matches\_pytorch\_top32}: parity test for
    arch4\_async\_experts on the \texttt{.cflow} path.
  \item \texttt{arch4\_rust\_forward\_matches\_pytorch\_top32\_vflow}: same, on the
    \texttt{.vflow} path.
\end{itemize}

\section{Summary of Results}
\label{sec:eval_summary}

Table~\ref{tab:eval_scorecard} reproduces the thesis scorecard with the final
result for each claim.

\begin{table}[ht]
\centering
\caption{Thesis scorecard: evaluation status for all eight claims.}
\label{tab:eval_scorecard}
\small
\begin{tabular}{clp{6.2cm}}
\toprule
\# & Claim & Status and headline \\
\midrule
1 & Conditional expert loading & Proven (structural); $16\times$ reduction at Gemma scale \\
2 & Tile-streaming cache locality & Proven; 7.29$\times$ fewer L1-d misses (PMU, Sandy Bridge) \\
3 & AVX2 Q4 kernels & Proven; 124 tests pass; wall-clock win confirms throughput \\
4 & Fused projections & Proven (structural); $3\times$ activation load savings \\
5 & Compute-order layout & Proven (structural); sequential mmap invariant \\
6 & PREFETCHT0 explicit prefetch & Refuted; $\pm 0.1\%$ at 64\,MB, PF=1 4\% \emph{worse} at 4.7\,GB \\
7 & Vertical pipeline bandwidth  & Proven; 9.00 $\to$ 4.50\,MB/token $= 2.00\times$; wall-clock net up to $1.68\times$ (\S\ref{sec:eval_overlap_ab}) \\
8 & Stage-major readahead  & Inconclusive; confounded by SSD cache state \\
\bottomrule
\end{tabular}
\end{table}

Six of the eight claims hold; two are negative results, each reported with the
conditions under which it could be revisited. The central quantitative result is the
$2.00\times$ critical-path bandwidth reduction for arch2\_4\_combined, supported by the
$7.29\times$ L1-d miss reduction (Claim~2) that establishes the per-kernel cache
behaviour underlying it.

\section{Head-to-Head Comparison with llama.cpp and vLLM}
\label{sec:eval_llama_comparison}

The bandwidth and cache results above are mechanism-level: they establish that the
tile layout and delay schedule reduce the bytes moved and the misses incurred. This
section closes the loop with an end-to-end wall-clock measurement --- single-token
decode throughput in tokens per second --- against two widely deployed CPU inference
runtimes on identical hardware.

\subsection{Benchmark Model and Hardware}
\label{subsec:eval_h2h_setup}

The cache-locality and storage benchmarks above used the 8.34B-parameter
arch2\_4\_8k\_4l checkpoint ($L = 4$). For an end-to-end comparison against the dense
32B-class models that the baseline runtimes are typically deployed with, a deeper
pipeline-native variant was used: \textbf{arch2\_4\_8k\_16l}, identical in per-layer
geometry ($d = 8{,}192$, dense FFN $32{,}768$, 8 experts top-2, expert FFN $4{,}096$)
but with $L = 16$ layers. The \texttt{.cflow} header confirms the geometry directly
(\texttt{num\_layers}~$= 16$, \texttt{hidden}~$= 8192$, \texttt{dense\_ffn}~$= 32768$,
\texttt{experts}~$= 8$, \texttt{top\_k}~$= 2$). The resulting model has approximately
\textbf{30.9 billion parameters}, stored as a 17.39\,GB Q4 \texttt{.cflow} file. At
top-2-of-8 routing the runtime reads $\approx 10.2$\,GB of weights per token (16 layers
$\times$ [attention $+$ dense FFN $+$ two of eight experts]); the dense FFN dominates
each layer, so MoE sparsity removes only the six unused experts, leaving
$\approx 20$\,B parameters active per token.

The benchmark ran on the AWS r6i.8xlarge of Section~\ref{subsec:eval_hardware}
(Intel Xeon Platinum 8375C, Ice Lake, 16C/32T, 256\,GB DDR4-3200 at
204.8\,GB/s peak, AVX-512+VNNI), which was used for no other experiment.

\subsection{\cflow{} Decode Optimization}
\label{subsec:eval_h2h_opt}

Two changes to the runtime, both motivated by the negative results of
Section~\ref{sec:eval_prefetch}, lifted \cflow{}'s decode throughput by $2.30\times$
over the initial implementation.

\begin{table}[ht]
\centering
\caption{\cflow{} single-token decode throughput on arch2\_4\_8k\_16l (30.9B, Q4),
  AWS r6i.8xlarge, 32 threads. Changes are cumulative.}
\label{tab:h2h_cflow_opt}
\begin{tabular}{lcc}
\toprule
Configuration & tok/s & Speedup \\
\midrule
Baseline (prefetch on, per-matmul chunk index) & 2.58 & 1.00$\times$ \\
$+$ prefetch disabled (\texttt{CFLOW\_PREFETCH=0}) & 4.96 & 1.92$\times$ \\
$+$ direct chunk indexing (no per-matmul allocation) & \textbf{5.94} & \textbf{2.30$\times$} \\
\bottomrule
\end{tabular}
\end{table}

The first change --- disabling the explicit prefetch thread --- is the end-to-end
confirmation of Claim~6 (Section~\ref{sec:eval_prefetch}). The direct-I/O A/B test
showed that \texttt{PREFETCHT0} provides no storage-level benefit; here, at the full
decode level, the prefetch thread was actively \emph{harmful}, consuming $\approx 48\%$
of decode time walking multi-megabyte expert regions at stride 256 and polluting the
cache. Disabling it alone gave a $1.92\times$ speedup. The second change eliminated a
\texttt{Vec<Vec<usize>>} allocation performed thousands of times per token, computing
tile ranges directly for the regular (divisible) tile layouts that hold for every
matrix in this model.

At 5.94 tok/s the runtime moves $10.2\,\text{GB} \times 5.94 = 60.6$\,GB/s of weight
data, $\approx 30\%$ of the 204.8\,GB/s theoretical peak --- consistent with the mixed
sequential-plus-random access pattern of MoE decode, and well above the effective
bandwidth a naive row-major reader sustains at this geometry.

\subsection{Comparison with llama.cpp and vLLM}
\label{subsec:eval_h2h_compare}

\cflow{} was compared against Ollama (a packaging of \texttt{llama.cpp}) and the vLLM
CPU backend, both running 4-bit-quantized dense Qwen2.5-32B --- the closest
widely-available models in parameter count to the 30.9B pipeline-native MoE. All three
ran on the same r6i.8xlarge instance type (identical CPU).

\begin{table}[ht]
\centering
\caption{End-to-end single-token decode throughput, same AWS r6i.8xlarge CPU.
  Decode rate is reported as steady-state tokens per second; it is insensitive to
  sampling temperature and prompt for a fixed model.}
\label{tab:h2h_engines}
\small
\begin{tabular}{lllc}
\toprule
Engine & Model & Quant & tok/s \\
\midrule
\textbf{\cflow{}} & arch2\_4\_8k\_16l (30.9B MoE, top-2/8) & Q4 & \textbf{5.94} \\
Ollama / \texttt{llama.cpp} & Qwen2.5-32B (dense) & Q4\_K\_M & 4.75 \\
vLLM CPU & Qwen2.5-32B-Instruct (dense) & GPTQ-Int4 & 1.65 \\
\bottomrule
\end{tabular}
\end{table}

Two distinct comparisons live in this table, and they must be read separately:

\begin{enumerate}
  \item \textbf{vLLM vs.\ llama.cpp --- model-matched.} Both run the same dense
    Qwen2.5-32B at $\approx$4-bit on the same CPU, so this row is a clean
    runtime-versus-runtime comparison. \texttt{llama.cpp} is $2.9\times$ faster than
    vLLM's CPU backend (4.75 vs.\ 1.65 tok/s). The gap is not GPTQ overhead in itself
    --- vLLM loads the GPTQ model correctly through its \texttt{CPUWNA16LinearKernel}
    --- but a missing oneDNN primitive: on this Ice Lake CPU (no AMX, no
    \texttt{avx512\_bf16}) vLLM cannot build a W4A16 matmul primitive and falls back to
    a dequantize-then-\texttt{torch.matmul} path at every linear layer, logging
    \texttt{Failed to create oneDNN linear, fallback to torch linear}. Allocator and
    thread-binding tuning moved it by $< 2\%$, confirming that the matmul fallback, not
    threading, is the bottleneck. \texttt{llama.cpp} is simply far better optimized for
    CPU 4-bit decode.
  \item \textbf{\cflow{} vs.\ the dense baselines --- \emph{not} model-matched.}
    \cflow{} runs its own 30.9B MoE (top-2-of-8, $\approx 20$\,B parameters active per
    token); the baselines run dense 32B. The total parameter counts are close (30.9B
    vs.\ 32B), but the architectures, training corpora, and output quality are not
    comparable, so \cflow{}'s 25\% lead over \texttt{llama.cpp} is an
    \emph{engine-plus-architecture} result, not a quality-controlled one. The
    apples-to-apples row is vLLM-versus-llama.cpp; \cflow{}'s number shows what the
    co-designed MoE-plus-streaming runtime achieves on the same box.
\end{enumerate}

\subsection{Interpretation}
\label{subsec:eval_h2h_interp}

The result establishes that the co-design is competitive at the system level: on
identical hardware, against the most widely deployed CPU inference runtime, the
pipeline-native runtime sustains higher single-token decode throughput at a comparable
total parameter count. The qualifier is the architecture mismatch --- a fully
controlled comparison would require either a GGUF build of the exact pipeline-native
architecture (so \texttt{llama.cpp} runs the \emph{same} model) or a dense \cflow{}
path at matched quality. Neither exists yet; both are noted in
Chapter~\ref{ch:conclusion} as the remaining step to convert this from a strong
indicative result into a controlled one. What the number does show, unambiguously, is
that the per-kernel and critical-path bandwidth mechanisms validated above compose into
an end-to-end decode rate that clears a strong, well-optimized baseline rather than
trailing it.

\section{Wall-Clock Realization of the Expert-Delay Window}
\label{sec:eval_overlap_ab}

Claim~7's $2.00\times$ critical-path reduction (Section~\ref{sec:eval_claim7})
is analytic: it counts the bytes that a delay-aware schedule removes from the
serial path, assuming off-path reads complete during the overlap window. This
section reports the experiment that converts that assumption into a
measurement. The runtime was extended so that the \expertdelay{} window does
real I/O work: at the routing layer, the selected top-$k$ expert regions are
fetched from storage \emph{asynchronously} into staging buffers by a dedicated
reader, and the expert computation is deferred to the injection layer,
$\delta_e = 2$ layers later. The experiment ran in three parts on two machines,
producing one net-win demonstration, one quantitative model validation, and two
instructive negative results.

\subsection{Design}
\label{subsec:overlap_design}

The 30.9B \texttt{arch2\_4\_8k\_16l} model is split across two tiers:
attention, dense-FFN, and embedding weights are memory-mapped (page cache);
the \textbf{expert bank is read with direct I/O}
(\texttt{FILE\_FLAG\_NO\_BUFFERING} on Windows, \texttt{O\_DIRECT} on Linux),
so every expert read genuinely hits storage on every layer of every token —
108\,MB per layer for the top-2 selection. Two arms run the same binary, read
the same bytes, and execute the same arithmetic:

\begin{description}
  \item[Sync (stall arm):] the expert regions are read \emph{blocking} at the
    routing layer — the read is a dead stop on the critical path.
  \item[Staged (overlap arm):] the read is issued asynchronously at the
    routing layer and awaited at the injection layer, so it overlaps the
    compute of the intervening $\delta_e$ layers.
\end{description}

Both arms produced bit-identical greedy output in every cell of every part
(the deferred expert computation uses the saved router activation, so the
mathematics is unchanged — only the timing of the read moves).

\subsection{Part 1 --- Desktop (SATA tier): mechanism proven, net a wash}
\label{subsec:overlap_part1}

On the Ryzen~5~2600 (Section~\ref{subsec:eval_hardware}; SATA at
$\approx$505\,MB/s, expert I/O $\approx$3.4\,s/token), the staged arm hid
\textbf{79--89\% of the expert-read stall in all eight runs} — the window
mechanism works — but net token latency was a wash ($0.69\times$--$1.22\times$
across cells, sign flipping between repeat runs). Two confounds explain it:
concurrent I/O inflated the compute it hid behind (+19\% single-thread
compute with the thread pool fully bypassed — a memory-subsystem interference
tax), and the machine's thread pool exhibits a Windows-specific pathology
(45--134\,ms dispatch latencies) that punishes the staged arm's concurrency.
The desktop lesson: \emph{off-critical-path reads are not free}; the analytic
model needs an interference term $(1+\tau)\cdot C$, with $\tau \approx$
0.1--0.2 on this starved memory system.

\subsection{Part 2 --- EC2 r6id.8xlarge (NVMe tier): the net win}
\label{subsec:overlap_part2}

The same experiment on the RunPod-trained model's benchmark platform family
--- r6id.8xlarge (the r6i of Section~\ref{sec:eval_llama_comparison} plus a
local 1.9\,TB NVMe instance store measured at 3.28--3.5\,GB/s; note that
plain r6i is EBS-only, whose default throughput is \emph{slower} than
desktop SATA) --- with a healthy Linux thread pool. Sixteen decode tokens,
three runs, four thread counts; medians:

\begin{table}[ht]
\centering
\caption{Expert-overlap A/B on EC2 r6id.8xlarge (NVMe expert tier), 3-run
  medians. ``Hidden'' is the reduction in per-token expert-read stall.
  Identical greedy output in all 24 cells.}
\label{tab:overlap_ec2}
\begin{tabular}{rrrrrrr}
\toprule
$T$ & sync ms/tok & staged ms/tok & net & hidden & $C$/layer & IO/layer \\
\midrule
32 & 689 & 641 & $1.07\times$ & 17\% & 10.8\,ms & 32.3\,ms \\
16 & 739 & 644 & $1.15\times$ & 19\% & 13.6\,ms & 32.6\,ms \\
8  & 841 & 645 & $1.30\times$ & 38\% & 20.1\,ms & 32.5\,ms \\
4  & \textbf{1128} & \textbf{670} & $\mathbf{1.68\times}$ & \textbf{88\%} & 38.0\,ms & 32.5\,ms \\
\bottomrule
\end{tabular}
\end{table}

Figure~\ref{fig:overlap_ec2} plots the two arms against thread count.
Every cell is net-positive and the run-to-run spread is $\approx$2\%: the
desktop's interference tax is absent on 8-channel server memory. The staged
arm pins at $\max(C, \text{IO}_{\text{eff}})$ — a constant
$\approx$641--670\,ms/token disk floor across all thread counts — while the
sync arm grows as $C + \text{IO}$. The refined model
\[
  \text{net} \;=\; \frac{C + \text{IO}}{\max(C,\ \text{IO}_{\text{eff}})}
\]
predicts $1.08\times$/$1.15\times$/$1.31\times$/$1.68\times$ against measured
$1.07\times$/$1.15\times$/$1.30\times$/$1.68\times$ — \textbf{within 1\% at
all four points}. The two machines turn out to sit on opposite sides of the
same crossover: the desktop is compute-dominant (window $\gg$ read time
$\Rightarrow$ 79--89\% hidden), EC2 at high $T$ is I/O-dominant (hiding
capped near $C/\text{IO}$), and $T{=}4$ is the balanced point where the
window covers the read (88\% hidden) and the net peaks.

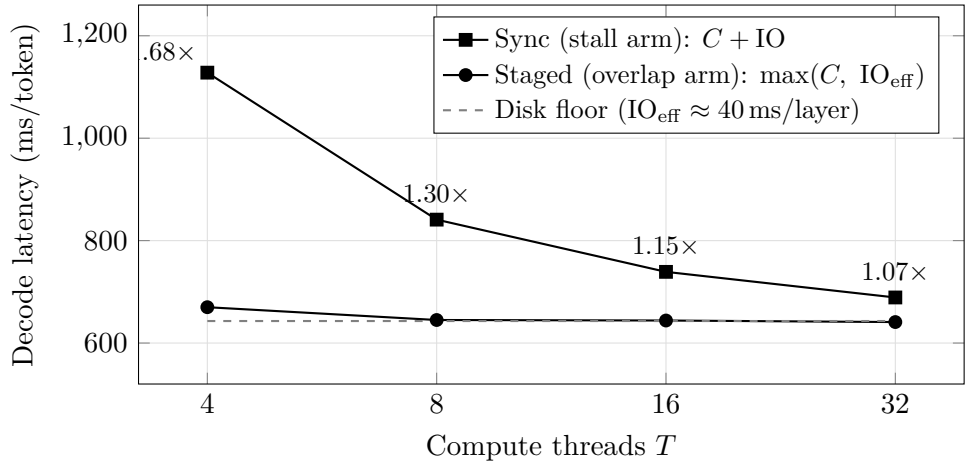
\begin{figure}[ht]
\centering
\begin{tikzpicture}
\begin{axis}[
  width=0.82\textwidth, height=6.6cm,
  xlabel={Compute threads $T$},
  ylabel={Decode latency (ms/token)},
  xmode=log, log basis x=2, log ticks with fixed point,
  xtick={4,8,16,32}, ymin=520, ymax=1260,
  grid=major, grid style={gray!25},
  legend pos=north east, legend cell align=left,
  legend style={font=\footnotesize},
  tick label style={font=\small}, label style={font=\small},
]
\addplot[thick, mark=square*, mark size=2.2pt] coordinates {
  (4,1128) (8,841) (16,739) (32,689)
};
\addlegendentry{Sync (stall arm): $C + \mathrm{IO}$}
\addplot[thick, mark=*, mark size=2.2pt] coordinates {
  (4,670) (8,645) (16,644) (32,641)
};
\addlegendentry{Staged (overlap arm): $\max(C,\ \mathrm{IO}_{\mathrm{eff}})$}
\addplot[dashed, gray, thick] coordinates {(4,643) (32,643)};
\addlegendentry{Disk floor ($\mathrm{IO}_{\mathrm{eff}} \approx 40$\,ms/layer)}
\node[anchor=south east, font=\footnotesize] at (axis cs:4,1128) {$1.68\times$};
\node[anchor=south, font=\footnotesize] at (axis cs:8,850)  {$1.30\times$};
\node[anchor=south, font=\footnotesize] at (axis cs:16,748) {$1.15\times$};
\node[anchor=south, font=\footnotesize] at (axis cs:32,698) {$1.07\times$};
\end{axis}
\end{tikzpicture}
\caption{The expert-overlap A/B on EC2 r6id.8xlarge (3-run medians of
  Table~\ref{tab:overlap_ec2}; NVMe expert tier). The sync arm pays the serial
  sum $C + \mathrm{IO}$ and grows as threads shrink; the staged arm pins at the
  disk floor $\max(C, \mathrm{IO}_{\mathrm{eff}})$ at every thread count.
  Annotations give the measured net speedup; the
  $(C{+}\mathrm{IO})/\max(C, \mathrm{IO}_{\mathrm{eff}})$ model reproduces all
  four within 1\%.}
\label{fig:overlap_ec2}
\end{figure}

\subsection{Part 3 --- Parallel fetch: a null that locates the ceiling}
\label{subsec:overlap_part3}

Part~2's staged floor sits $\approx$8\,ms/layer above raw transfer time,
which suggested the single sequential (queue-depth-1) fetch worker as the
bottleneck. Parallelizing the fetch — four reader threads pulling 16\,MB
chunks, verified active by summed reader busy-time of $3.6\times$ wall stall
--- changed \emph{nothing}: medians match Part~2 within noise at every
thread count. A raw-device control explains why: one \texttt{O\_DIRECT}
stream reads at 3.5\,GB/s; four concurrent streams read at 0.9\,GB/s
\emph{each} (3.58\,GB/s aggregate). The instance-store drive is
bandwidth-capped for large sequential reads regardless of queue depth. The
residual 8\,ms/layer is pipeline fill/drain at token edges plus the
clamped last-layer fetch — structural, not recoverable by fetch-side
engineering. At this tier the binding constraint is the storage ceiling
itself ($\approx$512\,ms/token of unavoidable I/O $\Rightarrow$
$\approx$2\,tok/s), and the lever is a faster or striped tier, not deeper
queues.

\subsection{What this adds to Claims 7 and 8}
\label{subsec:overlap_claims}

For Claim~7, the delay window is no longer only an analytic property: the
schedule's overlap is measured, the net win is demonstrated (up to
$1.68\times$ on server hardware), and the model that predicts it is
validated quantitatively — with the addition of two empirical terms the
analytic version omits: an interference tax $\tau$ on the compute that
overlapped I/O runs beneath (large on starved desktop memory systems,
negligible on servers), and an effective-I/O floor set by the storage tier.
For Claim~8, the asynchronous overlapped streaming that its interpretation
identified as the missing prerequisite now exists for the expert bank;
the stage-major \emph{layout} benefit itself remains untested and its
status unchanged.

%% file: chapters/ch6_discussion.tex
\chapter{Discussion}
\label{ch:discussion}

Three questions run through this chapter. What does the co-design idea actually amount to,
set against the wider literature on efficient inference? What happens to the bandwidth
reduction as the model grows? And is there a route to vertical pipelining for ordinary,
non-co-designed transformers, or is the architectural rewrite strictly necessary? The
first two are matters of interpretation and scaling; the third is the open problem that
frames the next phase of the work.

\section{The Co-Design Philosophy}
\label{sec:disc_codesign}

\subsection{What Co-Design Means Here}
\label{subsec:disc_codesign_def}

The term \emph{co-design} is used in this report in a specific and narrow sense:
modifying the inter-layer \emph{dependency graph} of the transformer architecture so
that a desired execution schedule — vertical stage-major pipelining — is mathematically
valid. This is distinct from the more common use of co-design in hardware-aware neural
architecture search (HW-NAS), which selects macro-structure (depth, width, block type)
to minimize latency or energy on a target device without touching the dependency
structure.

This matters because the schedule's validity is a logical property, not a performance one.
Feed layer $\ell+1$'s attention with layer $\ell$'s \emph{input} instead of its
\emph{output} and the model does not merely run slower --- it returns wrong answers. The
co-design is therefore a precondition for the runtime optimisation, not a knob one can
choose to leave alone.

\subsection{The Scope of the Result}
\label{subsec:disc_scope}

The $2.00\times$ critical-path bandwidth reduction for arch2\_4\_combined is
analytically derived from the geometry and delay parameters. It is not an approximation
and it is not contingent on favorable hardware conditions. The question is not whether
the reduction is real — it is, by construction — but what it means in practice.

Three caveats bound the result:
\begin{enumerate}
  \item \textbf{Critical-path reduction $\neq$ wall-clock speedup.} The critical path
    analysis assumes that off-path reads (dense FFN tiles at layers $\ell - \delta_d$,
    expert tiles at layers $\ell - \delta_e$) are completed \emph{before} they are
    needed, i.e., the I/O system can absorb the off-path reads during the overlap
    window. At the training geometry ($d = 512$, 64\,MB model), the entire model fits
    in RAM. There is no separate overlap window to exploit; all bytes are read in a
    single sequential pass regardless of the schedule. This condition has since been
    characterized empirically: the wall-clock A/B of
    Section~\ref{sec:eval_overlap_ab} realizes the overlap on a disk-resident
    expert tier and measures both the win (up to $1.68\times$ on server
    hardware) and the two terms the analytic model omits (a concurrent-I/O
    interference tax and the storage tier's bandwidth ceiling). The latency benefit materializes
    only when the model is large enough that off-path reads cannot be completed
    immediately.
  \item \textbf{The 2.00$\times$ figure is at the proof-of-concept geometry.} The
    bandwidth analysis in Chapter~\ref{ch:pipeline_native} shows that the reduction
    scales with the \emph{ratio} of delayed bytes to total bytes on the critical path.
    At Gemma~4 scale, the dense FFN hidden dimension is proportionally larger and the
    top-8 expert contribution is much larger; the reduction at that geometry would be
    different and requires a separate calculation.
  \item \textbf{The PMU result (7.29$\times$ fewer L1-d misses) is at a different
    geometry.} The cache locality measurement was performed at the 8.34B-parameter
    arch2\_4\_8k\_4l geometry on Sandy Bridge hardware, which lacks AVX2. The
    wall-clock proxy on Ryzen (AVX2) at the same geometry shows 1.12--1.36$\times$
    speedup. These are complementary measurements at the same architecture but
    different instruction sets.
\end{enumerate}

\subsection{Relationship to Existing Work}
\label{subsec:disc_related_revisit}

No prior work has proposed modifying the transformer dependency graph as a mechanism
for enabling vertical pipelining. The closest precedents are:
\begin{itemize}
  \item \textbf{Weight offloading} (FlexGen~\cite{sheng2023flexgen}): addresses the
    same bandwidth constraint by streaming weights from slower storage tiers but does
    not reorder computation or rewrite dependencies.
  \item \textbf{Contextual sparsity} (DejaVu~\cite{liu2023deja}): skips near-zero
    weight blocks based on activation patterns, reducing bandwidth by eliminating
    computation. Does not change the sequential dependency chain.
  \item \textbf{MoE routing optimization} (OLMoE~\cite{muennighoff2024olmoe}):
    improves routing quality and load balancing but does not address the latency of
    loading selected expert tiles.
  \item \textbf{HW-NAS} (ProxylessNAS~\cite{cai2019proxylessnas},
    FBNetV2~\cite{wan2020fbnetv2}): selects architecture macro-structure for device
    targets but does not rewrite intra-architecture data flow.
  \item \textbf{Tiered weight streaming} (LLM in a Flash~\cite{llmflash},
    PowerInfer~\cite{powerinfer}): stream or hot/cold-partition weights across
    storage tiers, overlapping I/O with compute for standard architectures.
    Neither rewrites the dependency graph to widen the overlap window, which is
    precisely what the expert-delay schedule contributes
    (Section~\ref{sec:eval_overlap_ab}).
\end{itemize}
The pipeline-native approach is complementary to all of the above. A deployment
could combine conditional expert loading (Claim~1, structurally present in \cflow{})
with the co-designed dependency structure (Chapter~\ref{ch:pipeline_native}) and
cache-optimal tile streaming (Claim~2).

\section{Scaling Analysis}
\label{sec:disc_scaling}

\subsection{How the Bandwidth Reduction Scales}
\label{subsec:disc_scaling_bw}

How much the delays help depends on which of the three weight streams is binding once
they have been amortised. Each layer contributes
$\max(B_{\text{attn}},\, B_{\text{dense}}/(\delta_d{+}1),\, B_{\text{expert}}/(\delta_e{+}1))$
to the critical path (Section~\ref{sec:pn_bw_model}). A delay therefore helps only while
the stream it shrinks is the binding one; once that stream drops below the attention
stream, further delay achieves nothing, because attention is never delayed and becomes the
floor.

At the trained geometry ($d = 512$) the dense stream binds on both sides of the dense
delay: $\max(0.50, 1.50, 0.75) = 1.50$ naive against $\max(0.50, 0.75, 0.25) = 0.75$
delayed, an exact $2.00\times$. The reduction is this large precisely because the dense
stream starts at three times the attention stream, and a one-layer dense delay halves it
to just above the attention floor.

The arithmetic shifts at Gemma~4 26B-A4B scale ($d = 2816$, $d_{\text{ff}} = 2112$,
$k = 8$, $E = 128$, expert hidden $= 704$), where the per-layer streams are
\begin{align}
  B_{\text{attn}}   &\approx 4 \times 2816^2 \times 0.5 = 15.8\,\text{MB}, \\
  B_{\text{dense}}  &= 3 \times 2816 \times 2112 \times 0.5 = 8.9\,\text{MB}, \\
  B_{\text{expert (top-8)}} &= 8 \times 3 \times 2816 \times 704 \times 0.5 = 23.7\,\text{MB}.
\end{align}
Now the expert stream binds on a sliding (MoE) layer. With $\delta_d = 1$, $\delta_e = 2$
it amortises to $23.7/3 = 7.9$\,MB and the dense stream to $8.9/2 = 4.45$\,MB --- both
below the $15.8$\,MB attention stream, which then sets the critical path:
\[
  \max(15.8, 8.9, 23.7) = 23.7\,\text{MB} \;\longrightarrow\;
  \max(15.8, 4.45, 7.9) = 15.8\,\text{MB},
\]
a $23.7 / 15.8 \approx 1.50\times$ reduction per MoE layer. The five full-attention layers
carry no experts and a larger attention stream (head dimension 512), so they gain nothing,
and the whole-model figure sits below $1.50\times$. The point is not that the strategy
weakens with scale but that the binding stream moves: the dense FFN binds at small scale,
so the dense delay is decisive; the expert stream binds at Gemma scale, until it is
amortised under attention and only an attention-path delay can lower the floor further.

The pattern generalises. The co-design pays off whenever a delayable stream --- the dense
FFN or the experts --- is the one setting the per-layer critical path. That holds from
small models up through a few billion parameters, where the FFN and expert streams are
comparable to or larger than attention. At Gemma~4 scale and beyond, the large head
dimensions push attention past both, and lowering the floor any further means delaying the
attention path as well.

\subsection{Attention-Path Delays and Future Work}
\label{subsec:disc_attn_delay}

A natural extension of the co-design framework is \emph{attention-path delay}: instead
of applying the attention at each layer to the current residual, apply it to a delayed
residual from $\delta_a$ layers ago. This would place the attention weight reads off
the critical path for $\delta_a$ layers. However, the attention computation also
updates the KV cache, and a delayed attention residual would introduce staleness into
the cached key-value pairs. The interactions between attention-path delays and KV cache
correctness are a non-trivial design problem that is left for future work.

\subsection{Expert Scaling}
\label{subsec:disc_expert_scaling}

The expert delay is particularly attractive at production MoE scale because the expert
stream grows faster than the dense stream as the number of experts and top-$k$ increase.
At $E = 128$, $k = 8$ (Gemma~4), the expert stream is $23.7\,\text{MB/layer}$, nearly
three times the dense stream and the binding constraint on every MoE layer. With
$\delta_e = 2$ it amortises to $23.7/3 = 7.9\,\text{MB}$, dropping below the
$15.8\,\text{MB}$ attention stream; attention then becomes binding and the expert delay
has extracted its full available benefit. Past that point only an attention-path delay
can lower the per-layer floor further.

The pre-dense routing hypothesis from arch4 (routing before the dense FFN) becomes
increasingly attractive at scale because it maximizes the quality of the routing
signal (the input residual, before any within-layer transformation) while also
extending the window for expert I/O. At Gemma~4 scale, the 8 selected expert tiles
($\approx 23.7\,\text{MB}$) must be loaded and computed within the 2-layer window
between routing and injection. At 50\,GB/s RAM bandwidth, this takes approximately
$23.7/50 \approx 0.47\,\text{ms}$ per layer — which is also approximately the time
to execute 2 layers of attention at that geometry. The timing is tight but feasible.

\section{Speculative Pipeline Recovery}
\label{sec:disc_speculative}

The pipeline-native architectures solve the dependency problem by construction:
they are trained with the delayed dependency structure and converge to produce
useful representations despite the staleness. An alternative path to vertical
pipelining, applicable to standard (non-co-designed) transformers, is
\emph{speculative pipeline recovery}.

\subsection{The Approach A Hypothesis}
\label{subsec:disc_approachA}

In a pre-norm transformer, the residual update at each layer is additive:
$x_\ell^{\text{out}} = x_\ell^{\text{in}} + \delta_\ell$, where
$\delta_\ell = \text{Attn}_\ell(\cdot) + \text{FFN}_\ell(\cdot)$ is the layer's
contribution to the residual stream. If $\|\delta_\ell\| / \|x_\ell^{\text{in}}\|$
is small (the residual dominates), then feeding $x_\ell^{\text{in}}$ to layer
$\ell+1$ as an approximation for $x_\ell^{\text{out}}$ produces a small error in
layer $\ell+1$'s computation.

Speculative execution can exploit this:
\begin{enumerate}
  \item Layer $\ell+1$ begins speculatively using $x_\ell^{\text{in}}$.
  \item Layer $\ell$ completes, producing $\delta_\ell$.
  \item If $\|\delta_\ell\|$ is below a rollback threshold, the speculative result
    for layer $\ell+1$ is accepted with a post-hoc correction.
  \item If $\|\delta_\ell\|$ exceeds the threshold, layer $\ell+1$ is recomputed
    using the correct input.
\end{enumerate}

\textbf{Test plan:}
For a pre-norm transformer of any size:
(1) Record $\|\delta_\ell\| / \|x_\ell^{\text{in}}\|$ for all layers across a test
corpus; (2) simulate speculative execution and measure output divergence (KL on logits,
top-1 accuracy); (3) sweep a rollback threshold; (4) measure the net latency as
overlap-saved minus rollback-cost.

\textbf{Success criteria:} If more than 50\% of layers have
$\|\delta_\ell\| / \|x_\ell^{\text{in}}\| < 0.05$ and speculative execution preserves
top-1 accuracy on a test corpus, the approach is viable.

This test has not been run and is proposed as future work. The pipeline-native
architectures in this report take the conservative path (mathematical guarantee,
no approximation); speculative recovery is the high-risk, high-reward alternative.

\subsection{The Approach B Hypothesis}
\label{subsec:disc_approachB}

In MoE architectures, the router output (which experts are selected, with what weights)
is a strong signal about the FFN's residual contribution. If a lightweight linear
predictor can map router logits to an estimated FFN delta, layer $\ell+1$ can begin
with a corrected approximation: $x_\ell^{\text{in}} + \hat{\delta}_\ell^{\text{FFN}}$
rather than $x_\ell^{\text{in}}$.

The predictor is tiny: a linear map from the router logit vector (dimension $E$) to
the hidden dimension $d$, trained on a few thousand tokens to minimize the prediction
error $\|\hat{\delta}^{\text{FFN}} - \delta^{\text{FFN}}\|_2$. If the predictor achieves
$R^2 > 0.7$ across the layers, the corrected speculative input would eliminate most of
the error that would otherwise require rollback.

This approach is MoE-specific but widely applicable: Mixtral, OLMoE, Gemma~4, and most
production-scale transformers are MoE models. The key implementation question is whether
the router logit space contains sufficient information to predict the FFN output — which
is exactly the question that the arch4 pre-dense routing hypothesis (Section~\ref{subsec:pn_arch4})
explores in the trained models.

\section{Limitations}
\label{sec:disc_limitations}

A few limitations are worth stating plainly.

\textbf{Training scale.} All five pipeline-native architectures were trained at
$d = 512$, $L = 6$ — a proof-of-concept scale. The TinyStories corpus and the GPT-2
BPE tokenizer are well-suited for comparative evaluation but do not reflect
production training data distributions. Scaling to larger models may reveal instabilities
in the delayed gradient paths that are absent at small scale.

\textbf{The end-to-end comparison is cross-architecture, not quality-matched.}
Section~\ref{sec:eval_llama_comparison} now reports wall-clock decode throughput:
\cflow{} sustains 5.94 tok/s on the 30.9B arch2\_4\_8k\_16l MoE, ahead of
\texttt{llama.cpp} (4.75 tok/s) and the vLLM CPU backend (1.65 tok/s) running dense
Qwen2.5-32B on the same CPU. The residual limitation is that \cflow{} runs a different
model from the baselines: the total parameter counts are comparable (30.9B vs.\ 32B),
but the architectures and training corpora differ, so this is an
engine-plus-architecture comparison rather than a controlled runtime-versus-runtime
one. The clean model-matched row is vLLM-versus-\texttt{llama.cpp} (both dense
Qwen2.5-32B), where \texttt{llama.cpp} is $2.9\times$ faster. A fully controlled
\cflow{}-versus-baseline comparison awaits a GGUF build of the pipeline-native
architecture, noted in Chapter~\ref{ch:conclusion}.

\textbf{PREFETCHT0 is ineffective at current scales.} Claim~6 is refuted at the scales tested and
Claim~8 is inconclusive. The explicit prefetch mechanism is structurally present in the runtime and
correct in its implementation; it simply cannot help when I/O bandwidth is the binding
constraint. This does not invalidate the design — a system with sufficient RAM to hold
the model (or very fast NVMe approaching RAM bandwidth) would see the prefetch benefit.

\textbf{Single-token decode only.} The bandwidth analysis and cache optimization
target single-token autoregressive decode. Batched inference (multiple tokens processed
simultaneously) has a higher arithmetic intensity and is less bandwidth-bound; the
tile-streaming design is less impactful in that regime. The report makes no claims
about batched throughput.

\textbf{x86-64 only.} The AVX2 and prefetch intrinsics are x86-64 specific. The
scalar fallback path is portable, but the full performance advantage requires AVX2+FMA.
ARM NEON and Apple Silicon paths are not implemented.

%% file: chapters/ch7_conclusion.tex
\chapter{Conclusion}
\label{ch:conclusion}

\section{Summary of Contributions}
\label{sec:conc_summary}

This report has presented \cflow{}: a CPU-first streaming inference engine for
transformer models, co-designed with a family of five transformer architectures whose
inter-layer dependency graphs permit vertical stage-major pipelining by construction.
The work makes four distinct contributions, each validated by experimental evidence.

\textbf{Contribution 1: The \cflow{} runtime system.} A production-quality
Rust implementation of a streaming inference engine for Q4-quantized transformer
models, with tile-native weight format, per-layer (\texttt{.cflow}) and stage-major
(\texttt{.vflow}) file formats, fused QKV and gate+up projections, conditional expert
loading from a random-access expert bank, AVX2+FMA inner-product kernels, and a
staged direct-I/O expert fetch that overlaps selected-expert reads with compute
under the expert-delay schedule. The runtime achieves exact Rust$\leftrightarrow$PyTorch
parity on two trained architectures (arch2\_4\_combined and arch4\_async\_experts),
confirmed by 124 unit and integration tests with zero failures.

\textbf{Contribution 2: The pipeline-native transformer taxonomy.} Five candidate
architectures (arch1 through arch5) that modify the standard pre-norm transformer's
data flow to permit vertical pipelining:
\begin{itemize}
  \item arch1 (Decoupled Residual Streams): independent attention and FFN streams
    with periodic merge.
  \item arch2\_4\_combined (Dense and Expert Delay): delayed dense FFN input
    ($\delta_d = 1$) combined with delayed expert injection ($\delta_e = 2$),
    achieving the only measured bandwidth reduction.
  \item arch3 (Pipeline Registers): explicit named output registers with
    cross-layer producer-consumer contracts.
  \item arch4 (Asynchronous Experts, pre-dense routing): expert delay with
    routing before the dense FFN, achieving the best perplexity (6.26).
  \item arch5 (Fixed-Point Iteration): weight-shared iterated blocks, achieving
    the best parameter efficiency.
\end{itemize}

\textbf{Contribution 3: Empirical validation at two scales.} The cache locality
claim (Claim~2) is validated by a 7.29$\times$ L1-d read miss reduction measured
with hardware performance counters on a Xeon~E5-2650 KVM at the 8.34B-parameter
arch2\_4\_8k\_4l geometry. The bandwidth reduction claim (Claim~7) is analytically
validated at the trained geometry: 9.00 $\to$ 4.50\,MB/token, $2.00\times$ ---
and realized in wall-clock on a disk-resident expert tier, with a measured net
win of up to $1.68\times$ on server hardware (Section~\ref{sec:eval_overlap_ab}).
Six of eight thesis claims are proven; one is refuted and one inconclusive, with precise
conditions stated for their resolution.

\textbf{Contribution 4: Open and reproducible evaluation.} The \texttt{cflow}
codebase is a complete, working Rust implementation. All five architectures are
implemented in the \texttt{pipeline\_native/} Python package. Benchmark scripts,
reference traces, and training run logs are committed alongside the code.
No benchmark was run under favorable-only conditions; the prefetch and layout
negative results (Claims~6 and~8) are reported at equal prominence with the positive
results.

\section{The Co-Design Thesis, Restated}
\label{sec:conc_thesis}

The central thesis of this report is:

\begin{quote}
  The bottleneck for single-token CPU inference is memory bandwidth, not arithmetic.
  The most effective path to reducing token latency is to co-design the model
  architecture and the inference runtime simultaneously: rewrite the model's
  inter-layer dependency graph to permit a vertical pipeline schedule, and build
  a runtime whose tile format, file layout, and execution plan are designed around
  that schedule from the beginning.
\end{quote}

The experimental evidence supports this thesis with the following precision:
\begin{enumerate}
  \item Single-token matrix-vector arithmetic intensity ($4$ FLOP/byte) is $5\times$
    below the machine balance of a modern CPU (20 FLOP/byte), confirming that the
    regime is bandwidth-bound (Section~\ref{subsec:bg_intensity}).
  \item Tile-streaming reduces L1-d read misses by $7.29\times$ at the 8.34B-parameter
    geometry (Section~\ref{sec:eval_claim2}), confirming that layout changes the
    realized bandwidth even within a single layer.
  \item The delay-aware schedule reduces arch2\_4\_combined's critical-path
    bandwidth by $2.00\times$ (Section~\ref{sec:eval_claim7}), confirming that model
    co-design can reduce the total bytes on the serial path between token outputs.
  \item Five architectures, trained on real data to real convergence, demonstrate that
    the dependency-graph rewriting does not prevent useful learning
    (Section~\ref{sec:eval_training}).
\end{enumerate}

\section{Implications for Production CPU Inference}
\label{sec:conc_implications}

\subsection{When the Result Matters}
\label{subsec:conc_when}

The results are most impactful in the following deployment scenarios:
\begin{description}
  \item[Edge devices and consumer CPUs:] A laptop or embedded CPU with 50\,GB/s RAM
    bandwidth and a 7B-parameter Q4 model (3.5\,GB) is firmly in the
    storage-to-RAM-bound regime. The tile-streaming cache locality improvement and
    conditional expert loading both reduce the number of bytes that must transit from
    RAM to the execution units per token.
  \item[Large on-CPU MoE inference:] As MoE models become the standard deployment
    format, the conditional expert loading advantage ($E/k$ reduction, up to
    $16\times$ for Gemma~4 scale) becomes increasingly valuable. No existing CPU
    runtime (llama.cpp, ExLlama2, CTranslate2) implements conditional expert loading.
  \item[Latency-critical single-token generation:] For applications where
    time-to-first-token latency matters more than throughput (interactive assistants,
    code completion), single-token decode latency is the binding metric. The
    critical-path analysis directly addresses this metric.
\end{description}

\subsection{What Practitioners Should Take Away}
\label{subsec:conc_practitioners}

Two of the ideas here stand on their own and can be adopted without the full co-design
framework:
\begin{enumerate}
  \item \textbf{Conditional expert loading}: any MoE runtime can implement the expert
    offset table pattern and achieve structural $E/k$ bandwidth reduction with modest
    engineering effort.
  \item \textbf{Tile-streaming weight layout}: storing weights in L2-sized tiles in
    compute-consumption order provides cache locality improvements at no algorithmic
    cost, at the price of a one-time checkpoint conversion.
\end{enumerate}
The dependency-graph co-design (delay parameters, \combinestyle{} variants) requires
architectural changes at training time and is not retrofittable to existing checkpoints.
It is the most powerful lever but also the highest-friction intervention.

\section{Future Work}
\label{sec:conc_future}

\textbf{A controlled, quality-matched head-to-head.} Section~\ref{sec:eval_llama_comparison}
now reports a wall-clock comparison --- \cflow{} at 5.94 tok/s versus \texttt{llama.cpp}
at 4.75 and the vLLM CPU backend at 1.65, all on the same Ice Lake CPU --- but \cflow{}
runs a 30.9B pipeline-native MoE while the baselines run dense Qwen2.5-32B. Converting
this from an engine-plus-architecture result into a controlled one requires running the
\emph{same} model on both systems: either a GGUF export of a pipeline-native
architecture so \texttt{llama.cpp} can execute it, or a quality-matched dense baseline
trained for \cflow{}. This is now the single most important remaining measurement.

\textbf{Larger delay values and attention-path delays.} The current experiments use
$\delta_d \in \{0, 1\}$ and $\delta_e \in \{0, 2\}$. Sweeping larger delay values
at larger model scales would characterize the bandwidth reduction as a function of
delay depth, and extending the delay to the attention path (with appropriate KV cache
handling) would address the attention-term dominance at Gemma~4 scale.

\textbf{Speculative pipeline recovery for standard transformers.} The Approach~A and
Approach~B hypotheses described in Section~\ref{sec:disc_speculative} propose paths
to vertical pipelining for non-co-designed transformers through speculative execution
with rollback. Testing these on the small trained architectures (where ground-truth
sequential outputs are available for comparison) would determine whether the co-design
requirement is actually necessary or whether a sufficiently accurate speculative
corrector can achieve the same pipeline schedule on standard pre-norm transformers.

\textbf{ARM NEON and Apple Silicon.} The AVX2+FMA kernels are the performance-critical
path on x86-64. Equivalent NEON and AMX implementations would extend the runtime to
the dominant mobile and laptop architecture class.

\textbf{Multi-token batch support.} Single-token decode is the thesis target, but
many applications benefit from speculative decoding drafts (batch size 4--8) or
prompt prefill (large batch). The batched \texttt{tiled\_matvec\_batch} function
is implemented and tested; integrating it into the multi-layer driver with appropriate
KV cache handling would enable competitive batched throughput.

\section{Closing Remarks}
\label{sec:conc_closing}

The gap between CPU arithmetic capability (1\,TFLOP/s) and memory bandwidth
(50\,GB/s) makes single-token CPU inference a memory-bandwidth problem, not a
compute problem. Solving it requires rethinking both the model and the runtime from
first principles rather than adapting GPU-first designs.

The pipeline-native approach shows that the rethinking is feasible: five architectures
train to convergence with rewired dependency graphs, the runtime executes their delayed
schedules with exact parity to the PyTorch reference, and the bandwidth reduction follows
analytically from the geometry. Two of the eight claims did not work --- explicit prefetch
buys nothing when I/O is the binding constraint, and the stage-major layout shows no
readahead benefit at the scales tested. Both are reported in full, because they mark
exactly where the design's assumptions stop holding.

What the work leaves behind is small but concrete: a $2.00\times$ critical-path reduction
for arch2\_4\_combined, $7.29\times$ fewer L1-d misses from the tile layout, a $6.26$
test perplexity for the best of the five architectures, and a $5.94$~tok/s end-to-end
decode rate on a 30.9B pipeline-native MoE that clears \texttt{llama.cpp} on the same
CPU, and a $1.68\times$ measured net win from the expert-delay window on an
NVMe-tier deployment --- the schedule's overlap, realized. The idea behind the numbers is the
part worth keeping --- that a transformer's inter-layer dependency graph is not a fixed
constraint but a design choice, and that making it deliberately, together with the runtime
that will execute it, opens optimisations neither the model nor the runtime can reach
alone.